\documentclass[lettersize,journal]{IEEEtran}

\IEEEoverridecommandlockouts

\IEEEoverridecommandlockouts
\usepackage[utf8]{inputenc}
\usepackage{color,setspace}
\usepackage{cite}
\usepackage{amssymb,amsfonts}
\usepackage{enumerate}
\usepackage{bbm}
\usepackage{graphicx}
\usepackage{epstopdf}
\usepackage[cmex10]{amsmath}
\usepackage{tasks}
\usepackage{enumitem}
\usepackage{bm}
\usepackage{xcolor}
\usepackage{xspace}
\usepackage{colortbl}
\usepackage{amsthm}
\usepackage{amsmath}
\usepackage{algpseudocode}
\usepackage[ruled,lined,commentsnumbered,linesnumbered]{algorithm2e}
\usepackage{subcaption}
\usepackage{caption}
\usepackage{url}

\def\BibTeX{{\rm B\kern-.05em{\sc i\kern-.025em b}\kern-.08em
T\kern-.1667em\lower.7ex\hbox{E}\kern-.125emX}}

\newtheorem{theorem}{\textit{Theorem}}

\newtheorem{corollary}{\textit{Corollary}}

\begin{document}

\title{Optimizing Information Freshness of IEEE 802.11ax Uplink OFDMA-Based Random Access}

\author{Jingwei Liu, Qian Wang, and He (Henry) Chen,~\IEEEmembership{Member,~IEEE,}
        % <-this % stops a space
\thanks{
The work of J. Liu is supported in part by the CUHK direct grant for research under Project 4055166.
The work of Q. Wang is supported in part by the National Natural Science Foundation of China under Grant 62401187.
The work of H. Chen is supported in part by RGC General Research Funds (GRF) under Project 14205020. The research work described in this paper was conducted in the JC STEM Lab of Advanced
Wireless
Networks for
Mission-Critical
Automation and
Intelligence funded by The Hong Kong Jockey Club Charities Trust.
% This article was presented in part at WiOpt 2023 \cite{10349810}.

J. Liu and H. Chen are with Department of Information Engineering, The Chinese University of Hong Kong, Hong Kong SAR, China 
(Email: \{lj020, he.chen\}@ie.cuhk.edu.hk).

Q. Wang is with School of Communication Engineering, Hangzhou Dianzi
University, China (email: qian.wang@hdu.edu.cn). 

% The authors thank Zhaorui Wang for his useful discussions on the network model.
}

}

% The paper headers
% \markboth{Journal of \LaTeX\ Class Files,~Vol.~14, No.~8, August~2021}%
% {Shell \MakeLowercase{\textit{et al.}}: A Sample Article Using IEEEtran.cls for IEEE Journals}

% \IEEEpubid{0000--0000/00\$00.00~\copyright~2021 IEEE}
% Remember, if you use this you must call \IEEEpubidadjcol in the second
% column for its text to clear the IEEEpubid mark.

\maketitle

\begin{abstract}
The latest WiFi standard, IEEE 802.11ax (WiFi 6), introduces a novel uplink random access mechanism called uplink orthogonal frequency division multiple access-based random access (UORA). While existing work has evaluated the performance of UORA using conventional performance metrics, such as throughput and delay, its information freshness performance has not been thoroughly investigated in the literature. This is of practical significance as WiFi 6 and beyond are expected to support real-time applications. This paper presents the first attempt to fill this gap by investigating the information freshness, quantified by the Age of Information (AoI) metric, in UORA networks. We establish an analytical framework comprising two discrete-time Markov chains (DTMCs) to characterize the transmission states of stations (STAs) in UORA networks. Building on the formulated DTMCs, we derive an analytical expression for the long-term average AoI (AAoI), facilitating the optimization of UORA parameters for enhanced AoI performance through exhaustive search. To gain deeper design insights and improve the effectiveness of UORA parameter optimization, we derive a closed-form expression for the AAoI and its approximated lower bound for a simplified
scenario characterized by a fixed backoff contention window and
generate-at-will status updates. By analyzing the approximated lower bound of the AAoI, we propose efficient UORA parameter optimization algorithms that can be realized with only a few comparisons of different possible values of the parameters to be optimized. 
% Simulation results validate our analysis and demonstrate that the AAoI achieved through the proposed parameter optimization algorithm closely approximates the optimal AoI performance obtained via exhaustive search.
% {\color{orange}Additionally, it outperforms the AAoI of the round-robin policy in dense and low-activity networks.}
Simulation results validate our analysis and demonstrate that the AAoI achieved through our proposed parameter optimization algorithm closely approximates the optimal AoI performance obtained via exhaustive search, outperforming the round-robin and max-AoI policies in large and low-traffic networks.
\end{abstract}


\section{Introduction}
\IEEEPARstart{I}{n} recent years, the deployment of IEEE 802.11ax (WiFi 6), the latest WiFi standard, has undergone a rapid expansion, marked by its increasing integration into various wireless networking applications.
% IEEE 802.11ax has revolutionarily adopted the Orthogonal Frequency Division Multiple Access (OFDMA) technology, which assigns subcarriers into multiple subsets, called resource units (RUs), enabling simultaneous data transmissions of various users.
IEEE 802.11ax has revolutionized WiFi communication by integrating the orthogonal frequency division multiple access (OFDMA) technology, enhancing spectral efficiency and network capacity 
%for high-density scenarios 
compared with the legacy IEEE 802.11 standards.
Specifically, OFDMA partitions orthogonal subcarriers into multiple wideband subsets known as resource units (RUs), each of which can carry an independent data stream, allowing simultaneous transmissions of multiple users \cite{4287203,4534773}. 
% The OFDMA-based Random Access (UORA) protocol, which is founded on OFDMA technology, represents a new mechanism for facilitating multi-user random access in the uplink as per the IEEE 802.11ax standard. 
The uplink orthogonal frequency division multiple access-based random access (UORA) protocol, which is founded on OFDMA technology, represents a new mechanism for facilitating multi-user random access in the uplink as per the IEEE 802.11ax standard.
% Prior research has investigated the performance of UORA networks using conventional metrics, including throughput and delay~\cite{10.1109/GLOCOM.2017.8254054, 7997340,7824740,6205590}.
Prior research has explored various UORA networks \cite{10.1109/GLOCOM.2017.8254054, 7997340,7824740,6205590,s19071540,9662442,9110744,10684959,yang2021utility,9040351,10068270}.
More specifically, authors of \cite{10.1109/GLOCOM.2017.8254054, 7997340,7824740,6205590,s19071540,9662442} evaluated the performance of UORA networks using conventional metrics, including throughput and delay.
Furthermore, studies in \cite{9110744,10684959,yang2021utility,9040351,10068270} proposed extended UORA mechanisms to enhance the conventional performances of the UORA networks.

% However, these works have not made efforts on assessing the information freshness performance of the UORA networks.
Thanks to the enhanced latency and capacity, WiFi 6 networks are anticipated to support time-sensitive applications, such as industrial Internet of Things, virtual reality, and smart health monitoring \cite{abdullah2015real,4550808}. 
For these applications, the timely transmission of fresh information is of paramount importance, as outdated information has little value. 
In this context, it is critically important to evaluate and optimize the information freshness performance of UORA networks. 
However, the literature indicates that traditional metrics such as delay and throughput are insufficient to characterize the freshness of information \cite{sun2022age,kosta2017age,6195689}. 
To address this gap, the age of information (AoI) metric has been introduced, providing a quantitative measure of information freshness for time-sensitive applications~\cite{kosta2017age,sun2022age}. 
Specifically, AoI is defined as the time elapsed since the generation of the most recently successfully received message at the destination. 
Significant research efforts have been dedicated to investigating and optimizing AoI in a variety of network contexts, see e.g., \cite{5984917,6195689,10.1145/3323679.3326520} and the references cited therein. 
% Nevertheless, to our best knowledge, the literature has not yet thoroughly examined the AoI of UORA networks.
Nevertheless, to the best of our knowledge, the existing literature has yet to thoroughly investigate the AoI of UORA networks, which could offer valuable design insights for optimizing UORA performance with respect to information freshness.

This paper represents the first attempt to address this gap by examining a symmetric wireless local area network that implements the UORA mechanism. 
In this network, multiple stations (STAs), each monitoring a physical process, aim to transmit their latest status updates to a common access point (AP) via a shared wireless channel that comprises multiple RUs available for random selection by the STAs. 
The arrival of status updates at the STAs is modeled as identically independent Bernoulli processes.
%A given number of RUs are avaliable allocated for the STAs to process their random access selections.
Our investigation reveals that evaluating the time-average AoI in UORA networks is a complex endeavor, primarily due to the intricate backoff mechanism employed by UORA. 
Specifically, each STA with a need to transmit uplink data must endure a randomly determined wait time within a specified backoff contention window before selecting an RU for transmission. Moreover, transmission failures can lead to an expansion of the backoff contention window range, further complicating the process. These dynamics critically affect the number of concurrent uplink transmissions, which in turn influences the rate of successful transmissions. Additionally, the interplay between the rate of successful transmissions and the packet arrival rate at the STAs impacts the overall demand for uplink transmissions within the network. This intricate interdependence among system components significantly complicates both the analysis and the optimization of AoI performance in UORA networks. To address these challenges, this paper makes two primary contributions, which are discussed as follows: 
%To the best of our acknowledgment, no comprehensive investigation of the AoI performance of the UORA networks has been implemented so far.

%The main contributions of this paper are outlined as follows.
%\begin{itemize}
Firstly, this paper formulates two discrete-time Markov chains (DTMCs) to establish an analytical framework for evaluating the long-term average AoI (AAoI) of UORA networks. 
% The first DTMC models the evolution of the number of STAs with status updates stored in their buffers, while the second DTMC represents the progression of the backoff states for each STA.
The first DTMC models the evolution of the number of STAs with status updates stored in their buffers, while the second DTMC captures the progression of the backoff states for each STA.
This allows in-depth analysis of the steady state of the considered system.
Building on this formulation, we employ a numerical approach to obtain the steady-state transmission success rate and the conditional probability that an STA gains access to the RUs after the completion of its initialized backoff waiting interval. Further analysis of the structure of the second DTMC, coupled with the obtained steady-state probabilities, enables us to derive an analytical expression for the AAoI. %More specifically, we calculate the expected time interval from the arrival of the first status update at an STA following its most recent successful transmission to its next successful transmission, as well as the expected square of this time interval, using a recursive algorithm.
Additionally, by leveraging the Bernoulli arrival of status updates, we approximate the AAoI of UORA networks. This systematic analysis facilitates the configuration of UORA network parameters to optimize AoI performance, using exhaustive search methods.
%Firstly, under a steady-state assumption, we formulate two discrete-time Markov chains (DTMCs) to be integrated into the analytical framework of the time-average AoI of the considered network. The first DTMC characterizes the evolution of the number of STAs with status updates stored in their buffers.    The second DTMC depicts how the backoff state of each STA evolves. On this basis, we can analyze the steady-state transmission success rate and the conditional probability that an STA has the right of access to the RUs after its initialized backoff waiting interval has elapsed by a numerical approach. After a thorough analysis of the structure of the second DTMC and utilizing the estimated steady-state probabilities, we can present the analytical expressions of significant components of a universal expression of the time-average AoI. More specifically, we derive the analytical expression of the expected time interval from the arrival of the first status update at an STA after its latest successful transmission to the next successful transmission of the STA, and that of the expected square of this time interval by a recursive algorithm. Additionally, based on the Bernoulli arrival process of the status update, we can approximate the AoI performance of the UORA network. Leveraging this groundwork, we can configure the UORA parameters to optimize the AoI performance through the exhaustive search.

Secondly, for deeper design insights, we analyze the AAoI of a simplified scenario characterized by a fixed backoff contention window and generate-at-will status updates. In this scenario, we derive a closed-form expression for a lower bound on the approximated AAoI. This lower bound is then used to guide the parameter optimization of UORA networks to enhance AoI performance. 
% Specifically, we transform the problem of minimizing this lower bound into an examination of the monotonicity properties of a continuous function. 
% Through comprehensive analysis, we establish the function’s monotonicity within the relevant domain, thus identifying the optimal parameters that minimize the AoI lower bound.
Specifically, we transform the problem of minimizing this lower bound into an examination of the monotonicity properties of a continuous function built upon the expression of the lower bound.
Subsequently, through comprehensive analysis, we establish the function’s
monotonicity within the relevant domain, thus identifying the optimal parameters that minimize the AoI lower bound.
Inspired by these findings, we propose an efficient UORA parameter optimization algorithm that can be realized with only a few comparisons of different possible values of the parameters to be optimized. We extend this strategy to accommodate the AoI-oriented optimization of general UORA networks with stochastic arrival of status updates. 
% {\color{blue}We note that a similar method has been adopted in \cite{10778312}, where the backoff contention window size of the IEEE 802.11 distributed coordination function (DCF) mechanism is fixed and optimized to replicate an AoI-optimal Bernoulli transmission policy.}
Simulation results validate the alignment of our analytical expression of the AAoI with simulation results obtained via Monte Carlo simulations. Furthermore, these simulations confirm the efficacy of our approximations and the robustness of the resulting parameter optimization algorithm.

\section{Related Work}
Considerable efforts have been dedicated to the study of AoI-oriented problems in ALOHA-like uplink random access networks, as evidenced by extensive literature, including \cite{9162973, 9785624, 9377549, chan2023lowpower, 8445979, 9791264, 8764468, 10138556, 10229041}. 
In \cite{9162973,9785624,9377549,8445979}, the AoI performance of slotted-ALOHA networks where sources generate packets at will was studied, while the authors in \cite{8445979,9791264,8764468,10229041} focused on that of the slotted-ALOHA networks with stochastic arrivals of status updates.
In particular, the authors in \cite{9162973,9785624,9377549} analyzed and optimized the AoI performance of the threshold-ALOHA networks, where each user transmits with a certain probability when its instantaneous AoI exceeds a predefined threshold. 
%The closed-form expression of the time-average AoI was achieved and optimized.
Moreover, the authors of \cite{10138556,chan2023lowpower} derived the analytical expression of the AoI performance of the frame slotted-ALOHA with reservation and data slots, followed by the associated optimization.
Our work offers analysis and optimization of the AoI in UORA networks involving an intricate backoff mechanism, which is not considered in ALOHA networks.

The AoI performance of another random access mechanism, carrier-sense multiple access (CSMA), has been explored in references \cite{5984917, 8493069, 10.1109/TNET.2020.2971350,10621330,10323421}. 
Particularly, in \cite{10.1109/TNET.2020.2971350}, the authors investigated the AoI performance of a CSMA network with randomly distributed parameters.
Based on the notion of stochastic hybrid systems (SHS), a closed-form expression of the network-wide average AoI was derived.
An AoI-oriented optimization problem was formulated and then converted to a convex problem, enabling efficient optimization.
Reference \cite{10621330} employed a novel SHS model incorporating the collision probability to study the AoI performance of a tagged node in a practical CSMA network including dense background nodes.
The AoI-optimal traffic arrival rate of the tagged node was analytically found on the basis of the creative model.
However, the optimization objective and the parameter to be optimized in \cite{10621330} differ from those in our work.
The authors of \cite{10323421} proposed a distributed policy called Fresh-CSMA based on the backoff mechanism of the CSMA technology to optimize AoI in single-hop wireless networks.
Fresh-CSMA was proven to match the centralized scheduling decisions of the max-weight policy, which is near-optimal, with a high probability in the same network state.
The authors also showed that Fresh-CSMA with a realistic setting performs comparably to the max-weight policy.
Nevertheless, despite the use of backoff mechanisms in both CSMA and UORA networks, the analytical framework developed for the AoI in CSMA networks cannot be directly applied to our work. This is due to the fundamental differences in the transmission processes between CSMA and UORA networks.

Optimizing the AoI performance in legacy WiFi networks has garnered increasing attention in recent literature. For instance, the authors in \cite{10000608} developed an optimization strategy based on queuing analysis of AoI in WiFi networks. Additionally, a deep learning approach for channel condition estimation aimed at reducing AoI in WiFi networks was introduced in \cite{9973486}. Furthermore, several studies, including \cite{9522228, 10228860, 10349886}, have developed and implemented AoI-based transmission schemes by adapting legacy IEEE 802.11 standards. Specifically, \cite{10228860} describes the implementation of an application layer middleware designed to tailor IEEE 802.11 networks to the requirements of time-sensitive devices. In \cite{10349886}, the authors proposed an AoI-optimized protocol stack that significantly enhances the AoI performance of WiFi devices. Nevertheless, these studies did not address the most recent updates to the IEEE 802.11 standards.

\section{System Model and Preliminaries}\label{sec:system}

We consider a basic service set (BSS) that implements the uplink orthogonal frequency division multiple access (OFDMA)-based random access (UORA) mechanism specified in the IEEE 802.11ax protocol. 
This BSS comprises an AP and $N$ time-sensitive STAs, which monitor their associated physical processes (e.g., the speed and position of mobile robots) and aim to transmit status updates timely to the AP via the uplink. 
We assume a stochastic arrival model for the generation of these status updates. 
Each STA is equipped with an integer OFDMA backoff (OBO) counter and maintains a size-1 buffer that stores only the most recently arrived status update. 
To provide a foundation for our analysis, we next provide a concise primer on the IEEE 802.11ax UORA mechanism. 
For a more comprehensive understanding, readers are referred to the 802.11ax amendment \cite{9442429}, which offers more protocol details.

\subsection{A Primer on IEEE 802.11ax UORA}
The procedure of the 802.11ax UORA transmission is shown in Fig.~\ref{UORA_TX}. Specifically, the AP regularly broadcasts a trigger frame (TF), including information about the number of resource units (RUs) allocated for random access, denoted by $L$.
% In this work, the amount of RUs $L$ is assumed to be fixed.
After receiving the TF, each STA with a status update in its buffer activates its OBO counter immediately if the OBO counter is offline. 
The OBO counter is initialized uniformly in the range of $0$ to the OFDMA contention window (OCW).
% which has an initial value $OCW_{min}$.
Subsequently, each active OBO counter is decremented by $L$ if its current count exceeds $L$; otherwise, it is reset to $0$. If the active OBO counter remains greater than $0$ after being decremented by $L$, it is maintained until the STA receives the next TF.
% In the following period of the HE (high-efficiency) trigger-based physical layer protocol data unit (PPDU), each STA with the OBO counter reaching $0$ uniformly selects one RU to transmit its status update in the buffer.
After a short interframe space (SIFS), each STA with the OBO counter reaching $0$ uniformly selects one RU to transmit its status update in the buffer within the following period of the high-efficiency (HE) trigger-based physical layer protocol data unit (PPDU).
% Those STAs with OBO counters equal to $0$, thus having random access rights, are referred to as \textit{active STAs} hereafter.
For simplicity, we consider the collision channel model. That is, if multiple STAs select the same RU to transmit, a collision occurs, causing all transmissions through that RU to fail.
Conversely, the AP can correctly receive the status update transmitted via an RU selected only by one STA.
% Then, the AP sends a multi-sTA blockACK (M-BA) to acknowledge which STAs have just successfully transmitted their status updates.
% The acknowledged status updates would be discarded from their buffers. 
% In this work, we define the duration between the beginnings of two consecutive TFs as a time slot, indexed by $t\in\{1,2,\cdots,T\}$, which is regarded as the atomic unit of time.
Next, the AP sends a Multi-STA BlockAck (M-BA) after another SIFS to acknowledge which STAs have just successfully transmitted their status updates.
The acknowledged status updates would be discarded from their buffers.
After a distributed coordination function interframe space (DIFS), the AP broadcasts the next TF for another transmission opportunity.
We define the duration between the beginnings of two consecutive TFs as a time slot, indexed by $t\in\{1,2,\cdots,T\}$.
By Fig. \ref{UORA_TX}, the length of a time slot $T_{slot}$ can be expressed as
\begin{equation*}
    T_{slot} = T_{TF}+T_{SIFS}+T_{payload}+T_{SIFS}+T_{ACK}+T_{DIFS},
\end{equation*}
where $T_{TF}$ denotes the transmission time of the TF, $T_{SIFS}$ is the duration of the SIFS, $T_{payload}$ is the transmission time of the HE trigger-based PPDU, $T_{ACK}$ is the duration of the M-BA, and $T_{DIFS}$ denotes the duration of DIFS.
In this work, we assume that the AP can perfectly maintain the length of time slots, and regard a time slot as the atomic unit of time.
The status update arrival at each STA in each time slot follows an independent and identically distributed (i.i.d.) Bernoulli process with an arrival rate of~ $\lambda$.\footnotemark{}
\footnotetext{In the continuous-time model, the packet arrival rate $\lambda$ (also known as the saturated rate) is calculated by $1-\exp(-\lambda_PT_{slot})$, where $\lambda_P$ denotes the Poisson arrival rate of the status updates in the sense of continuous time.}

\begin{figure}%[t]
	\centering    \includegraphics[width=0.48\textwidth]{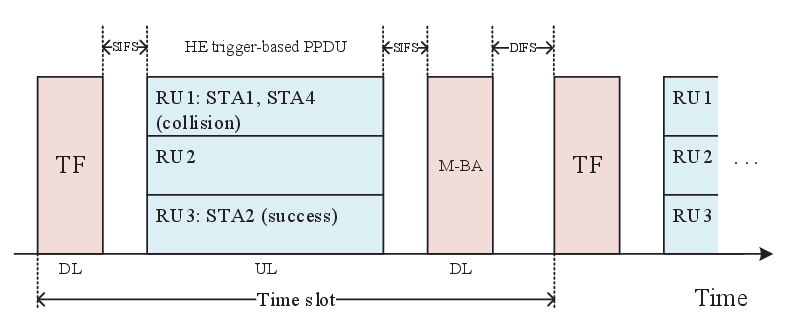}
	\caption{Transmission process in IEEE 802.11ax UORA networks, where SIFS stands for Short Interframe Space.}
	\label{UORA_TX}
 \vspace{-1em}
\end{figure}

The evolution of the OBO counter follows an exponential backoff mechanism \cite{9442429}.
Let $x$ denote the backoff level of the OBO counter, and $OCW_x$ denote the value of OCW in backoff level $x$.
When an OBO counter is activated by its STA, we have $x=0$ and $OCW_0=OCW_{min}$.
After the M-BA is broadcast, the STAs aware of their successful transmissions reset OCW to $OCW_{min}$ and take their OBO counters offline, while the STAs whose transmissions failed update OCW and initialize their OBO counters.
Specifically, each failed transmission of the STA increases the backoff level of its OBO counter by $1$ and updates its OCW by
\begin{equation}\label{OCWeolve}
    OCW_x=\min(2\times OCW_{x-1}+1,OCW_{max}).
\end{equation}
The values of $OCW_{min}$ and $OCW_{max}$ are determined by 
\begin{equation}\label{OCWmin}
    OCW_{min}=2^{EOCW_{min}}-1,
\end{equation}
\begin{equation}\label{OCWmax}
    OCW_{max}=2^{EOCW_{max}}-1,
\end{equation}
where $EOCW_{min}$ and $EOCW_{max}$ are integer parameters notified to the STAs through the beacon frame sent by the AP.
% In IEEE 802.11ax, we have the predefined range $0\le EOCW_{min}\le EOCW_{max}\le 1023$.
In IEEE 802.11ax, we have the predefined range $0\le EOCW_{min}\le EOCW_{max}\le 7$.
Clearly, we have the highest backoff level, denoted by $m$.
For notation brevity, we define
\begin{equation}
    W_x\triangleq OCW_x+1.
\end{equation}
By \eqref{OCWeolve}, \eqref{OCWmin}, and \eqref{OCWmax}, we have $W_0=OCW_{min}+1$, $W_m=OCW_{max}+1$, and $W_x=2^xW_0$.

%Hereafter, we refer to those  STAs with OBO counters equal to $0$, thus having random access rights to transmit, as the \textit{accessing STAs},  those STAs with status updates in the buffers as the \textit{active STAs}, and those STAs with empty buffers and offline OBO counters as the \textit{idle STAs}, respectively.
We now define the STAs with different statuses in UORA networks. 
STAs with OBO counters set to $0$, which grants them the right to transmit, are referred to as \textit{accessing STAs}. Those STAs that currently have status updates in their buffers are termed \textit{active STAs}. 
Lastly, STAs characterized by empty buffers and inactive OBO counters are designated as \textit{idle STAs}.

\subsection{AoI as Information Freshness Metric}
The AoI metric, initially proposed in \cite{5984917}, is adopted to assess the information freshness at the AP. 
Given that a symmetric network configuration is assumed, our analysis can be concentrated on the AoI for a typical STA. 
Moving forward, we denote the AoI at the typical STA and the system time of the most recent status update stored in the buffer of this STA in slot $t$ as $\Delta(t)$ and $\delta(t)$, respectively. 
The mathematical expression for the evolution of $\delta(t)$ is defined~as
\begin{equation}\label{LAoIevolve}
    \delta(t+1)=
    \begin{cases}
        0,& \text{if a status update arrives at the STA}\\
        & \text{in slot $t+1$,}\\
        \delta(t)+1,& \text{otherwise.}
    \end{cases}
\end{equation}
Moreover, by the definition of the average AoI, $\Delta(t)$ evolves in the form of
\begin{equation}\label{AoIevolve}
    \Delta(t+1)=
    \begin{cases}
        \delta(t)+1,& \text{if a status update of the STA is}\\
        & \text{received by the AP in slot $t$,}\\
        \Delta(t)+1,& \text{otherwise.}
    \end{cases}
\end{equation}

% The performance metric we aim to evaluate for the typical STA is the long-term average AoI (AAoI), denoted by $\overline{\Delta}$. This metric is mathematically defined as follows:
% \begin{equation}\label{AAoI}
%     \overline{\Delta}\triangleq \lim_{T\to\infty}\frac{1}{T}\sum^T_{t=1}\Delta(t).
% \end{equation}
% % Note that once the system reaches a steady state, the AoI for different STAs evolves identically due to the symmetric setup of the network.
% Due to the symmetric setup of the network, the AoI for different STAs evolves identically in the long term.
% Consequently, we can use \eqref{AAoI} to represent the AoI performance across the entire network.
The performance metric we aim to evaluate for the typical STA is the time average expected AoI (AAoI), denoted by $\overline{\Delta}$. This metric is mathematically defined as follows:
\begin{equation}\label{AAoI}
    \overline{\Delta}\triangleq \lim_{T\to\infty}\frac{1}{T}\mathbb{E}\left[\sum^T_{t=1}\Delta(t)\right].
\end{equation}
Due to the symmetric setup of the network, the steady-state distributions of AoI are the same for different STAs, implying that AAoIs of different STAs are identical.
Consequently, we can use \eqref{AAoI} to represent the AoI performance across the entire network. 
Thanks to the time-slotted feature of the considered UORA network, we can express $\overline{\Delta}$ as~\cite{8648195}
\begin{equation}\label{AoIUE}
    \overline{\Delta}=\mathbb{E}[S_{\tau-1}]+\frac{\mathbb{E}[X^2_{\tau}]}{2\mathbb{E}[X_{\tau}]}-\frac{1}{2},
\end{equation}
where $X_{\tau}$ denotes the time interval between the $(\tau-1)$th and $\tau$th receptions of status updates from the typical STA, and $S_{\tau-1}$ represents the service time of the $(\tau-1)$th received status update, which is the time interval between the generation time and the reception time of this status update. 
In the subsequent section, we will introduce an analytical framework designed to thoroughly analyze each expectation term presented in (\ref{AoIUE}).

\section{Analytical Framework}
\begin{figure}[t]
	\centering
    \includegraphics[width=0.47\textwidth]{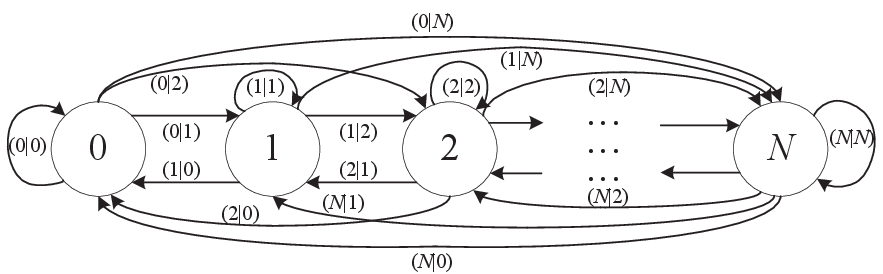}
	\caption{The DTMC of the number of active STAs, i.e., $\mathcal{M}_1$, where $(i|j)$ stands for $P_{i,j}$.}
    % \vspace{-1em}
	\label{MCN}
  \vspace{-1em}
\end{figure}
{
This section introduces an analytical framework designed to analyze the long-term AAoI in UORA networks. The primary challenge in AAoI analysis involves characterizing the UORA process and its impact on the AoI evolution of each STA, particularly with stochastic arrivals of status updates. Specifically, the OBO counter for each STA activates within the UORA process only if there is a status update in its buffer—indicating the STA is active. This activation affects the number of STAs contending for $L$ Resource Units (RUs), thereby complicating the AoI evolution across all active STAs. Moreover, status preemption during the UORA process further alters the AoI dynamics. We note that the existing frameworks for performance analysis of UORA networks, which are limited to saturation conditions and do not account for status preemption \cite{7997340,7824740}, are not suitable for our analysis.} 
%Since all STAs in the BSS share the common $L$ RUs, the states of all STAs are entangled together.

{
For the entangled AoI evolution of all STAs, it is feasible to model the instantaneous AoI evolution using a multi-dimensional discrete-time Markov chain (DTMC). However, this approach encounters significant challenges due to the curse of dimensionality, particularly as the number of STAs, $N$, increases. 
To address this issue, we adopt a critical approximation in which each status update transmission is assumed to experience a constant and independent probability of successful transmission (and correspondingly, collision probability), irrespective of the number of retransmissions. This approximation is widely employed in the literature for performance analysis of random access protocols (see e.g., \cite{840210,6205590}), and has been shown to achieve high accuracy, particularly when the number of STAs is large. 
With this approximation in place, we can construct two DTMCs to characterize the state transitions respectively: one for the number of active STAs at each time slot, and another for the state transitions experienced by each STA within the UORA process. However, these two DTMCs are not independent but are mutually coupled. Specifically, the number of active STAs affects the potential number of concurrent uplink transmissions in the UORA process. Conversely, the state transitions within the UORA process influence the distribution of the number of active STAs. This interdependence makes the analysis complex and non-trivial.} In this section, we carefully analyze the relationships among the abovementioned factors when constructing the two DTMCs. We then calculate the steady-state distribution to analyze the successful probability of each transmission attempt, and finally apply a recursive method to analyze the AAoI.
%Since the UORA process is not equivalent to the AoI evolution process, we make use of the recursive feature of the DTMC chain for the UORA process to analyze the EAoI.

%To address this challenge, in this section we first carefully analyze the relationship between these factors to construct the two DTMCs. 
%Then, we calculate the steady-state distribution to analyze the successful probability of each transmission attempt and the EAoI. Since the UORA process is not equivalent to the AoI evolution process, we make use of the recursive feature of the DTMC chain for the UORA process to analyze the EAoI.}
\subsection{DTMC Construction}%Relevant Probabilities}

Denote by $q$ and $\rho$ the successful transmission probability of an accessing STA and the stationary channel accessing probability of an active STA, respectively. Recall the definitions of active STA and accessing STA, the active STA considers the status arrival process and the status transmission outcome during the UORA process, while accessing STA characterizes the STA state in the UORA process. As such, the value of $q$ depends on the number of active STAs and is necessary for characterizing the UORA process. In this sense, we begin with constructing a DTMC $\mathcal{M}_1$ to characterize the evolution of the number of active STAs for a given $\rho$.
The structure of $\mathcal{M}_1$ is shown in Fig. \ref{MCN}.
The number of active STAs is defined as the DTMC state of $\mathcal{M}_1$, with the stationary distribution of states being denoted by $\boldsymbol{\mu}=[\mu_i]^N_{i=0}$, where $\mu_i$ denotes the steady-state probability that the BSS has $i$ active STAs. 
The state transition matrix of $\mathcal{M}_1$ is denoted by $\mathbf{P}\triangleq [P_{i,j}]_{0\le i,j\le N}$, where $P_{i,j}$ is the state transition probability that the BSS has $j$ active STAs in the current time slot under the condition that the BSS has $i$ active STAs in the last time slot, which is given by 
\begin{equation}
    P_{i,j}=\sum^{\min\{i,L\}}_{s=\max\{0,i-j\}}A^{j-i+s}_{N-i+s}(\lambda)D^s_i(\rho,L),
\end{equation}
where $A^{k}_{n}(\lambda)$ denotes the probability that $k$ out of $n$ idle STAs have new status updates arriving in one slot, given by 
% $A^{k}_{n}(\lambda)=\binom{n}{k}\lambda^k(1-\lambda)^{n-k}$,
\begin{equation}
    A^{k}_{n}(\lambda)=\binom{n}{k}\lambda^k(1-\lambda)^{n-k},
\end{equation}
% and $D^s_i(\rho,L)$ denotes the probability that $s$ out of $i$ accessing STAs have successful transmissions.
% % , and $\rho$ is the steady-state probability that an STA is accessing an RU given that the STA is active.
% Specifically, we have 
% % $D^s_i(\rho,L)=\sum^i_{j=s}C^j_i(\rho)T^s_j(L)$,
% \begin{equation}
%     D^s_i(\rho,L)=\sum^i_{j=s}C^j_i(\rho)T^s_j(L),
% \end{equation}
% where $C^j_i(\rho)$ represents the probability that $j$ out of $i$ active STAs are accessing STAs, given by $C^j_i(\rho)=\binom{i}{j}\rho^j(1-\rho)^{i-j}$,
% % \begin{equation}
% %     C^j_i(\rho)=\binom{i}{j}\rho^j(1-\rho)^{i-j},
% % \end{equation}
% and $T^s_j(L)$ denotes the probability that $s$ out of $j$ accessing STAs transmit successfully.
% $T^s_j(L)$ is equivalent to the probability of randomly assigning $j$ objects to $L$ units where $s$ units are assigned a single object, which has been studied in \cite{1095866}.
% In reference to \cite[Eq. 6]{1095866}, we have
% \begin{equation}
%     T^s_j(L)=\frac{(-1)^sL!j!}{L^js!}\sum^{\min\{L,j\}}_{n=s}\frac{(-1)^n(L-n)^{j-n}}{(n-s)!(L-n)!(j-n)!}.
% \end{equation}
and $D^s_i(\rho,L)$ denotes the probability that $s$ out of $i$ active STAs have successful transmissions.
Specifically, we have 
% $D^s_i(\rho,L)=\sum^i_{j=s}C^j_i(\rho)T^s_j(L)$,
\begin{equation}
    D^s_i(\rho,L)=\sum^i_{g=s}C^g_i(\rho)T^s_g(L),
\end{equation}
where $C^g_i(\rho)$ represents the probability that $g$ out of $i$ active STAs are accessing STAs, given by $C^g_i(\rho)=\binom{i}{g}\rho^g(1-\rho)^{i-g}$,
% \begin{equation}
%     C^j_i(\rho)=\binom{i}{j}\rho^j(1-\rho)^{i-j},
% \end{equation}
and $T^s_g(L)$ denotes the probability that $s$ out of $g$ accessing STAs transmit successfully.
% $T^s_g(L)$ is equivalent to the probability of randomly assigning $g$ objects to $L$ units where $s$ units are assigned a single object, which has been studied in \cite{1095866}.
More specifically, $T^s_g(L)$ is the probability that each of $j$ accessing STAs uniformly selects one of $L$ RUs, and precisely $s$ of those $L$ RUs are chosen by only one STA, leading to successful transmissions on those RUs.
Obviously, this probability is equivalent to the probability of randomly assigning $g$ objects to $L$ units where $s$ units are assigned a single object, which has been studied in \cite{1095866}.
In reference to \cite[Eq. 6]{1095866}, we have
\begin{equation}
    T^s_g(L)=\frac{(-1)^sL!g!}{L^gs!}\sum^{\min\{L,g\}}_{h=s}\frac{(-1)^h(L-h)^{g-h}}{(h-s)!(L-h)!(g-h)!}.
\end{equation}
% Then, according to the balance equations of the DTMC \cite{norris1998markov}, we have 
% \begin{equation}\label{ls}
%     \begin{bmatrix}
%         \boldsymbol{\mu}^T\\
%         1
%     \end{bmatrix}=
%     \begin{bmatrix}
%         \mathbf{P}^T\\
%         \mathbf{1}
%     \end{bmatrix}\boldsymbol{\mu}^T,
% \end{equation}
% which can be manipulated into a linear system, and can easily derive
% % the following proposition by Bayes theorem.
% % \begin{proposition}
% % Given $\boldsymbol{\mu}$, $\rho$, $N$, and $L$, the transmission success rate of an STA is given by
% \begin{equation}\label{TSR}
%     q=\sum^{N-1}_{a=0}\frac{(a+1)\mu_{a+1}}{\sum^{N-1}_{i=0}(i+1)\mu_{i+1}}\sum^a_{b=0}C^b_a(\rho)(1-\frac{1}{L})^b
%     \end{equation}
% % \end{proposition}
% by Bayes theorem and the law of total probability.
Then, according to the balance equations of the DTMC \cite{norris1998markov}, we have 
\begin{equation}\label{ls}
    \begin{bmatrix}
        \boldsymbol{\mu}^T\\
        1
    \end{bmatrix}=
    \begin{bmatrix}
        \mathbf{P}^T\\
        \mathbf{1}
    \end{bmatrix}\boldsymbol{\mu}^T,
\end{equation}
which can be manipulated into a linear system.
Additionally, we can derive the expression of $q$, given in the following theorem.
\begin{theorem}\label{T0}
    Given $\boldsymbol{\mu}$, $\rho$, $N$, and $L$, the transmission success rate of an STA is expressed as
    \begin{equation}\label{TSR}
        q=\sum^{N-1}_{a=0}\frac{(a+1)\mu_{a+1}}{\sum^{N-1}_{i=0}(i+1)\mu_{i+1}}\sum^a_{b=0}C^b_a(\rho)(1-\frac{1}{L})^b.
    \end{equation}
\end{theorem}
\begin{proof}
    See Appendix \ref{appA0}.
\end{proof}

To solve \eqref{ls} and calculate \eqref{TSR}, we delve into the investigation of the crucial probability $\rho$.
To that end, we construct another DTMC $\mathcal{M}_2$ to characterize the evolution of the state of the OBO counter of an STA.
Based on the backoff mechanism, we define two categories of states for $\mathcal{M}_2$.
The first category of the state is denoted by $\left\langle0\right\rangle$, which represents the STA is idle, and thus the OBO counter of the STA is offline.
The second category of states is denoted by bi-dimensional tuples $\left\langle x,y \right\rangle$, where $x$ defined in Section~\ref{sec:system} denotes the backoff level of the OBO counter, $y$ denotes the count of the OBO counter after the STA updates the counter at the beginning of a slot.
Recall that the OBO count is decremented by $L$ immediately after the counter is initialized.
Hence, given $x$, we have $y\in\{0,1,\cdots,\max(0,W_x-L-1)\}$.
The structure of $\mathcal{M}_2$ is depicted in Fig. \ref{MC}, where $\alpha_x\triangleq\lfloor\frac{W_x-1}{L}\rfloor$ denotes the quotient, and $\beta_x=W_x-1-\alpha_xL$ is the remainder of $\frac{W_x-1}{L}$.

\begin{figure}[t]
	\centering
    \includegraphics[width=0.45\textwidth]{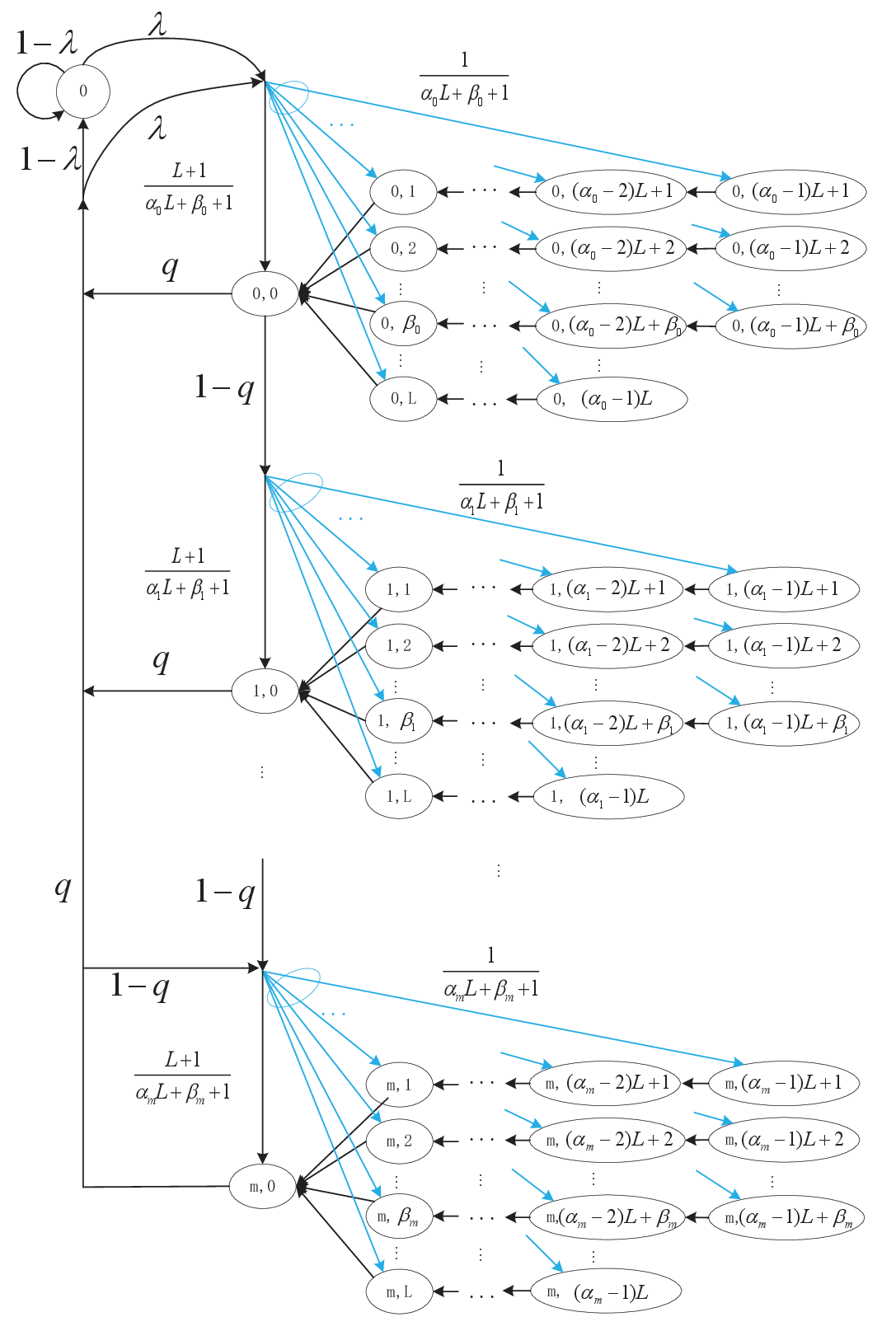}
	\caption{The DTMC of the state of the OBO counter, i.e., $\mathcal{M}_2$.}
	\label{MC}
 \vspace{-1em}
\end{figure}

We denote the steady-state probabilities of states $\left\langle0\right\rangle$ and $\left\langle x,y \right\rangle$ by $\pi_0$ and $\pi_{x,y}$, respectively.

By Fig. \ref{MC}, we have 
\begin{equation}\label{rho}
    \rho=\frac{\sum^m_{x=0}\pi_{x,0}}{\sum^m_{x=0}\sum^{W_x-L-1}_{y=0}\pi_{x,y}}.
\end{equation}
Clearly, $\pi_{x,0}=(1-q)\pi_{x-1,0}$ for $0<x<m$, and $\pi_{m,0}=(1-q)\pi_{m-1,0}+(1-q)\pi_{m,0}=\frac{(1-q)^m}{q}\pi_{0,0}$.
As such, we can obtain $\sum^m_{x=0}\pi_{x,0}=\frac{\pi_{0,0}}{q}$.
% \begin{equation}\label{pia0}
%     \sum^m_{a=0}\pi_{a,0}=\frac{\pi_{0,0}}{q}.
% \end{equation}

% We remark that a similar DTMC, representing a special case of $\mathcal{M}_2$, was employed in \cite{10.1109/GLOCOM.2017.8254054}.
% In the DTMC of \cite{10.1109/GLOCOM.2017.8254054}, the transitions between state $\left\langle x,y \right\rangle$ are equivalent to that in $\mathcal{M}_2$.
We remark that a similar DTMC, representing a special case of $\mathcal{M}_2$ with $\lambda=1$, was employed in \cite{10.1109/GLOCOM.2017.8254054}.
In the DTMC of \cite{10.1109/GLOCOM.2017.8254054}, the transitions between state $\left\langle x,y \right\rangle$ are equivalent to that in $\mathcal{M}_2$ since they are not affected by $\lambda$.
By applying the results in \cite[Eqs. 7, 9]{10.1109/GLOCOM.2017.8254054}, we have
\begin{equation}\label{sumpiab}
\begin{split}
    &\sum^m_{x=0}\sum^{W_x-L-1}_{y=0}\pi_{x,y}\\
    &=
    \begin{cases}
        \pi_{0,0}\frac{q\sum^{m-1}_{x=0}H_x(\frac{1-q}{2})^x+H_m(\frac{1-q}{2})^m+W_0}{W_0q}, &\text{if }m>0,\\
        \pi_{0,0}\frac{H_0+W_0}{W_0}, &\text{otherwise},
    \end{cases}
\end{split}
\end{equation}
where 
% $H_x=-\frac{L}{2}\left\lfloor\frac{W_x-1}{L}\right\rfloor^2+\left(W_x-1-\frac{L}{2}\right)\left\lfloor\frac{W_x-1}{L}\right\rfloor$.
\begin{equation}
    H_x=-\frac{L}{2}\left\lfloor\frac{W_x-1}{L}\right\rfloor^2+\left(W_x-1-\frac{L}{2}\right)\left\lfloor\frac{W_x-1}{L}\right\rfloor.
\end{equation}
Substituting \eqref{sumpiab} together with $\sum^m_{x=0}\pi_{x,0}=\frac{\pi_{0,0}}{q}$ into \eqref{rho} yields 
% the expression of $\rho$.
\begin{equation}\label{CAR}
    \rho=\frac{W_0}{q\sum^{m-1}_{x=0}H_x(\frac{1-q}{2})^x+H_m(\frac{1-q}{2})^m+W_0}
\end{equation}
when $m>0$, and $\rho=\frac{W_0}{H_0+W_0}$ when $m=0$.
With \eqref{ls}, \eqref{TSR}, and \eqref{rho}, we can calculate the values of $q$ and $\rho$ by utilizing numerical approaches.

\subsection{AAoI Analysis}\label{EAoIAnaylsis}%Time Intervals}
With the previous analysis on the constructed DTMCs, we are ready to analyze each expected term (i.e., $\mathbb{E}[X_{\tau}]$, $\mathbb{E}[X^2_{\tau}]$, and $\mathbb{E}[S_{\tau-1}]$) in the expression of AAoI  $\overline{\Delta}$ given in~\eqref{AoIUE}.
%Based on the former analysis, we focus on analyzing the expected terms: $\mathbb{E}[X_{\tau}]$, $\mathbb{E}[X^2_{\tau}]$, and $\mathbb{E}[S_{\tau-1}]$ to derive $\overline{\Delta}$ according to~\eqref{AoIUE}.

% According to stochastic arrival nature of status update, we have $X_{\tau}=V_{\tau}+K_{\tau}$, where $V_{\tau}$ denotes the waiting time of the $\tau$th transmission of an STA, i.e., the elapsed time slots since the $(\tau-1)$th packet reception of the STA until the next status update arrival at the STA, and $K_{\tau}$ is the time interval from the arrival of the next status update at the STA to the $\tau$th successful transmission of the STA.
\begin{figure}[t]
	\centering    \includegraphics[width=0.48\textwidth]{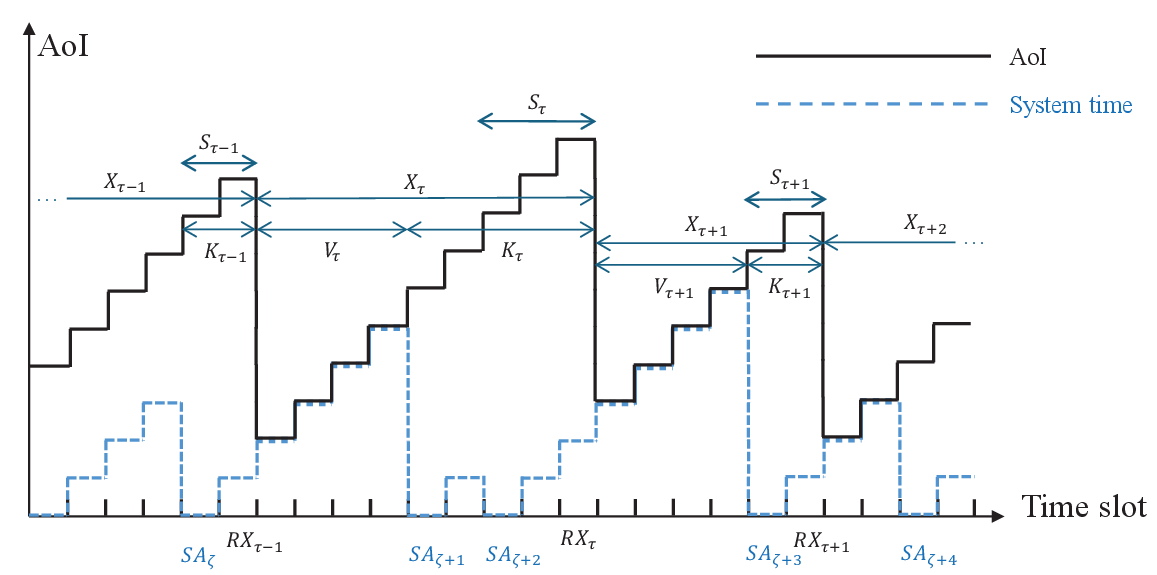}
	\caption{The AoI evolution of an STA in the considered model.}
	\label{AoI_Evolution_UORA}
 \vspace{-1em}
\end{figure}
To intuitively present the processes within $S_{\tau-1}$ and $X_{\tau}$, we depict the AoI evolution of an STA in the considered time-slotted model in Fig. \ref{AoI_Evolution_UORA}, where $RX_{\tau}$ denotes the time slot of the $\tau$th reception of the status update packet at the AP, and $SA_{\zeta}$ denotes the time slot of the $\zeta$th status update arrival at the STA. 
As shown in Fig. \ref{AoI_Evolution_UORA}, according to the stochastic arrival nature of status update at the STA, we have $X_{\tau}=V_{\tau}+K_{\tau}$.
Specifically, $V_{\tau}$ denotes the waiting time of
the $\tau$th reception of the status update packet at the AP, representing the elapsed time slots after the $(\tau-1)$th packet reception at the AP until the first new status update arrival at the STA, i.e, 
\begin{equation}
    V_{\tau}\triangleq\min_{SA_{\zeta'}>RX_{\tau-1}}\{SA_{\zeta'}\}-RX_{\tau-1}-1.
\end{equation}
Moreover, $K_{\tau}$ denotes the time interval from the first new arrival of a status update at the STA after the $(\tau-1)$th packet reception at the AP till the $\tau$th packet reception at the AP, mathematically,
\begin{equation}
    K_{\tau}\triangleq RX_{\tau}-\min_{SA_{\zeta'}>RX_{\tau-1}}\{SA_{\zeta'}\}+1.
\end{equation}
Obviously, $V_{\tau}$, $K_{\tau}$, and $S_{\tau}$ are independent of each other and independent across $\tau$.
We hence omit $\tau$ for brevity hereafter unless specified, and naturally obtain
\begin{equation}
    \mathbb{E}[X]=\mathbb{E}[V]+\mathbb{E}[K],
\end{equation}
\begin{equation}
    \mathbb{E}[X^2]=\mathbb{E}[V^2]+\mathbb{E}[K^2]+2\mathbb{E}[V]\mathbb{E}[K].
\end{equation}
Note that $V$ only depends on the Bernoulli arrival process of the status update, implying that $V$ follows the geometric distribution.
In light of this, we have
\begin{equation}
    \mathbb{E}[V]=\sum^\infty_{t=1}t\lambda(1-\lambda)^{t-1}=\frac{1}{\lambda}-1,
\end{equation}
\begin{equation}
    \mathbb{E}[V^2]=\sum^\infty_{t=1}t^2\lambda(1-\lambda)^{t-1}=\frac{\lambda^2-3\lambda+2}{\lambda^2}-1.
\end{equation}
To proceed towards the derivation on $\mathbb{E}[K]$ and $\mathbb{E}[K^2]$, we let $R_x$ denote a variable representing the consecutive time slots elapsed from the initial entry of the OBO counter of an STA into level $x$ until the STA successfully transmits a packet.
Based on the structure of $\mathcal{M}_2$ shown in Fig. \ref{MC}, we achieve the following theorem.
\begin{theorem}\label{T1}
    Given $m$, $L$, $W_x$, and $q$, for $x=m$, we have
    % \begin{equation}\label{ERmCF}
    %     \mathbb{E}[R_m]=\frac{1}{q}\left[\frac{\alpha_m(\alpha_m+1)}{2}\frac{L}{W_m}+\frac{(\alpha_m+1)\beta_m+1}{W_m}\right],
    % \end{equation}
    \begin{equation}\label{ERmCF}
        \mathbb{E}[R_m]=\frac{1}{q}\mathbb{E}[U_m],
    \end{equation}
    % \begin{equation}\label{ERmSqCF}
    % \begin{split}
    %     \mathbb{E}[R^2_m]& \!=\!\frac{1}{q}\!\left[\frac{\alpha_m(\alpha_m\!+\!1)(2\alpha_m \!+\!1)}{6}\frac{L}{W_m}\!+\!\frac{(\alpha_m\!+\!1)^2\beta_m\!+\!1}   {W_m}\right.\\
    %     &\left.+\frac{2(1-q)}{q}\!\left[\frac{\alpha_m(\alpha_m\!+\!1)}{2}\frac{L}{W_m}\!+\!\frac{(\alpha_m\!+\!1)\beta_m\!+\!1}{W_m}\right]^2\right],
    % \end{split}
    % \end{equation}
    \begin{equation}\label{ERmSqCF}
        \mathbb{E}[R^2_m]=\frac{1}{q}\mathbb{E}[U_m^2]+\frac{2(1-q)}{q^2}\mathbb{E}[U_m]^2,
    \end{equation}
    and, for $0\le x<m$, we have
    % \begin{equation}\label{ERaCF}
    % \begin{split}
    %     &\mathbb{E}[R_x]=
    %     \frac{\alpha_x(\alpha_x+1)}{2}\frac{L}{W_x}+\frac{(\alpha_x+1)\beta_x+1}{W_x}+(1-q)\mathbb{E}[R_{x+1}],
    % \end{split}
    % \end{equation}
    \begin{equation}\label{ERaCF}
    \begin{split}
        &\mathbb{E}[R_x]=
        \mathbb{E}[U_x]+(1-q)\mathbb{E}[R_{x+1}],
    \end{split}
    \end{equation}
    % \begin{equation}\label{ERaSqCF}
    %     \begin{split}
    %         \mathbb{E}[R_x^2]&=\frac{\alpha_x(\alpha_x+1)(2\alpha_x+1)}{6}\frac{L}{W_x}+\frac{(\alpha_x+1)^2\beta_x+1}{W_x}\\
    %         &+2(1-q)\mathbb{E}[R_{x+1}]\left[\frac{\alpha_x(\alpha_x+1)}{2}\frac{L}{W_x}\!+\!\frac{(\alpha_x+1)\beta_x+1}{W_x}\right]\\
    %         &+(1-q)\mathbb{E}[R_{x+1}^2].
    %     \end{split}
    % \end{equation}
    \begin{equation}\label{ERaSqCF}
        \begin{split}
            \mathbb{E}[R_x^2]&=\mathbb{E}[U_x^2]
            +2(1-q)\mathbb{E}[R_{x+1}]\mathbb{E}[U_x]
            +(1-q)\mathbb{E}[R_{x+1}^2],
        \end{split}
    \end{equation}
    where $U_x$ denotes the waiting duration until the next transmission attempt after $x$th consecutive failed transmission (i.e., $x$th consecutive backoffs) and the entry into backoff level $x$.
    More specifically,
    \begin{equation}
        \mathbb{E}[U_x]=\frac{\alpha_x(\alpha_x+1)}{2}\frac{L}{\alpha_xL+\beta_x+1}+\frac{(\alpha_x+1)\beta_x+1}{\alpha_xL+\beta_x+1},
    \end{equation}
    and
    \begin{equation}
        \mathbb{E}[U_x^2]=\frac{\alpha_x(\alpha_x+1)(2\alpha_x+1)}{6}\frac{L}{W_x}+\frac{(\alpha_x+1)^2\beta_x+1}{W_x}.
    \end{equation}
\end{theorem}
\begin{proof}
    See 
    Appendix \ref{appA}.
\end{proof}
Built upon \textit{Theorem} \ref{T1}, we can obtain $\mathbb{E}[K]$ and $\mathbb{E}[K^2]$ by $\mathbb{E}[K]=\mathbb{E}[R_0]$ and $\mathbb{E}[K^2]=\mathbb{E}[R^2_0]$, respectively.

Following this, our focus shifts to compute $\mathbb{E}[S]$.
% Drawing upon the definition of the service time, the value of $S_{\tau}$ depends on $X_{\tau}$, indicating that we need to investigate the probability distribution of $X$.
% However, we find that the corresponding distribution analysis is intractable due to the backoff mechanism of the OBO counter.
The service time of the $\tau$th received status at the AP, denoted as $S_{\tau}$, is influenced by the state of the OBO counter at the time of status arrival. If the OBO counter is active, its current state will affect the distribution of $S_{\tau}$. Additionally, it is crucial to ensure that no new status arrivals occur during this process. A further concern is that the model $\mathcal{M}_2$ for the UORA process does not account for status preemption when a new status arrives. This omission renders the recursive methods for calculating $\mathbb{E}[K]$ and $\mathbb{E}[K^2]$ inapplicable. To simplify the computation of $\mathbb{E}[S]$, we use the steady state of $\mathcal{M}_2$. We assume that each active STA accesses RUs in every time slot following a Bernoulli process with probability $\rho$. This assumption is logical, as it maintains that the conditional probability of an active STA becoming an accessing STA is $\rho$ in each slot once the system reaches the steady state. 
% Under this assumption, 
% the expression for $\mathbb{E}[S]$ can be derived in the same manner as \cite[Eq. 13]{10138556}, and is expressed as
% \begin{equation}
%     \mathbb{E}[S]=\sum^\infty_{i=1}\frac{i\Pr\{s=i\}}{\sum^{\infty}_{i=1}\Pr\{s=i\}},
% \end{equation}
% where $\Pr\{s=i\}$ denotes the probability that a status update is transmitted successfully after $i$ slots since its arrival without preemption.
% By the assumption, we have $\Pr\{s=i\}=(1-\lambda)^{i-1}(1-\rho q)^{i-1}\rho q$, where $(1-\lambda)^{i-1}$ is the probability that no preemption occurs before the reception of the current status update, and $(1-\rho q)^{i-1}\rho q$ represents the probability that the current status update is received by the AP in the $i$th slot after the status update's arrival.
Under this assumption, the expression for $\mathbb{E}[S]$ can be derived in the same manner as \cite[Eq. 13]{10138556}, and is expressed as
\begin{equation}
    \mathbb{E}[S]=\sum^\infty_{I=1}\frac{I\Pr\{T_S=I\}}{\sum^{\infty}_{I=1}\Pr\{T_S=I\}},
\end{equation}
where $\Pr\{T_S=I\}$ denotes the probability that a status update is transmitted successfully after $I$ slots since its arrival without preemption.
By the assumption, we have $\Pr\{T_S=I\}=(1-\lambda)^{I-1}(1-\rho q)^{I-1}\rho q$, where $(1-\lambda)^{I-1}$ is the probability that no preemption occurs before the reception of the current status update, and $(1-\rho q)^{I-1}\rho q$ represents the probability that the current status update is received by the AP in the $I$th slot after the status update's arrival.
After some manipulations, we have
\begin{equation}\label{approS}
    \mathbb{E}[S]=\frac{1}{\lambda(1-\rho q)+\rho q}.
\end{equation}
Consequently, with the obtained expressions for $\mathbb{E}[K]$, $\mathbb{E}[K^2]$, $\mathbb{E}[V]$, $\mathbb{E}[V^2]$, and $\mathbb{E}[S]$, we can attain an analytical expression of $\overline{\Delta}$ by substituting these expressions into \eqref{AoIUE}.
{Furthermore, according to the operation principles of the UORA mechanism, we can arrive at the following corollary.
\begin{corollary}\label{AoINoChange}
    Given the number of STAs $N$, the number of available RUs $L$, and the arrival rate of status updates $\lambda$, the AAoI $\overline{\Delta}$ has the same expression when $W_m\le L+1$.
\end{corollary}}
 \begin{proof}
  %See 
  % Appendix B of \cite{CitationLiu}.
  %Appendix \ref{appC1}.
  According to the backoff process, if an STA enters level $x$, the OBO counter would be initialized as $OBO\le W_x-1 \le W_m-1\le L$. 
    % indicating that the STA always chooses to transmit when it is active.
    Thus, the backoff state of the STA must transit to $\left\langle x,0 \right\rangle$ after the STA enters level $x$, yielding that $\pi_{x,y}=0$ for $y>0$.
    Substituting $\pi_{x,y}=0$ for $y>0$ into \eqref{rho}, we have $\rho=1$.
    Consequently, we can obtain the same $\mathbb{E}[K]$, $\mathbb{E}[K^2]$, $\mathbb{E}[V]$, $\mathbb{E}[V^2]$, and $\mathbb{E}[S]$ in these contexts, leading to the same expression of $\overline{\Delta}$.
  \end{proof}
The above corollary will be applied in the subsequent section to optimize UORA parameters, thereby reducing AAoI. Moreover, the key approximation and assumption utilized to derive an analytical expression for AAoI will be validated through simulations presented in Section~\ref{sec:simulation}.
    % According to the backoff process, if an STA enters level $x$, the OBO counter would be initialized as $OBO\le W_x-1 \le W_m-1\le L$. 
    % % indicating that the STA always chooses to transmit when it is active.
    % Thus, the backoff state of the STA must transit to $\left\langle a,0 \right\rangle$ after the STA enters level $x$, yielding that $\pi_{a,b}=0$ for $b>0$.
    % Substituting $\pi_{a,b}=0$ for $b>0$ into \eqref{rho}, we have $\rho=1$.
    % Consequently, we can obtain the same $\mathbb{E}[K]$, $\mathbb{E}[K^2]$, $\mathbb{E}[V]$, $\mathbb{E}[V^2]$, and $\mathbb{E}[S]$ in these contexts, leading to the same expression of $\overline{\Delta}$.
    % See Appendix \ref{appC1}.
% \end{proof}
% }
\section{AAoI Optimization}

With the analytical expression of $\overline{\Delta}$ obtained in the previous section, 
we can formulate the following AAoI optimization problem to optimize the OBO backoff parameters in UORA networks with given $N$, $L$, and $\lambda$:
\begin{equation}\label{eq:optimization}
    (W_0^*,m^*)=\arg \min_{W_0,m}\overline{\Delta}.
\end{equation}
Note that $EOCW_{min}=\log_2(W_0)$ and $EOCW_{max}=EOCW_{min}+m$.
We note that an exhaustive search is required to find the optimal solution to the problem outlined in (\ref{eq:optimization}). 
% This necessity arises because the parameters $q$ and $\rho$ must be determined using numerical methods. 
This necessity arises because the parameters $q$ and $\rho$ are coupled together, as described by equations \eqref{TSR} and \eqref{CAR}, necessitating their determination through numerical methods. Furthermore, $\mathbb{E}[U_x]$ and $\mathbb{E}[U_x^2]$ in $\overline{\Delta}$ are difficult to analyze precisely due to their piecewise multivariate nature.
To overcome this challenge and derive deeper design insights, we aim to develop a closed-form expression for the AAoI, which will facilitate a more efficient optimization of the UORA parameters. 
To this end, this section begins by considering a simplified UORA network with ``generate-at-will'' status updates, where it is assumed that each STA has a status update arrival in every time slot, i.e., $\lambda = 1$. We then extend these results to the general case with arbitrary $\lambda$ towards the end of this section. 
%in calculating the estimated AAoI for different UORA parameters makes the general optimization methods no longer applicable to optimize the UORA parameters, we can still apply the to optimize the UORA parameters with the analytical AAoI expression in \eqref{AoIUE}. 
%Further investigation of the optimization has been left as a future work and we then focus on a special configuration to gain more insight of the UORA design.
% In this sense, we cannot express $\overline{\Delta}$ in an analytical form, indicating that a theoretical analysis on $\overline{\Delta}$ is an open question for us.
% That is, we lack design insight to devise a more effective AoI-oriented optimization algorithm.
% Therefore, further investigation is left as a future work, and we are motivated to explore and devise an efficient optimization scheme for the UORA network under a special setting.

% Recall that $q$ and $\rho$ are achieved by numerical methods given a general network setting $(N,L,W_0,m,\lambda)$.
% In this sense, we cannot express $\overline{\Delta}$ in an analytical form, indicating that a theoretical analysis on $\overline{\Delta}$ is an open question for us.

%We first propose a design concept, then investigate the AoI performance based on the design concept, followed by formulating an efficient optimization algorithm for this configuration.
%\subsection{Design Concept}
We now detail our critical observation essential for the AoI-oriented optimization of UORA networks: the OBO backoff level parameter $m$ should be minimized (i.e., set to 0) to enhance AoI performance.
% denote by $U_x$ the waiting duration until the next transmission attempt after $x$th consecutive failed transmission (i.e., $x$th consecutive backoffs) and the entry into backoff level $x$.
% After some analysis, we obtain the following expression regarding the expectation of $U_x$
% To proceed, we calculate its partial derivatives of $\mathbb{E}[U_x]$ given in Theorem \ref{T1} and have
% \begin{equation}\label{pd1}
%     \frac{\partial \mathbb{E}[U_x]}{\partial \alpha_x}\!=\!\dfrac{L^2\alpha_x^2\!+\!\left(2L\beta_x+2L\right)\alpha_x+2\beta_x^2+\left(2-L\right)\beta_x\!-\!L}{2\left(L\alpha_x+\beta_x+1\right)^2},
% \end{equation} 
% \begin{equation}\label{pd2}
%     \frac{\partial \mathbb{E}[U_x]}{\partial \beta_x}=\dfrac{\alpha_x\left(L\alpha_x+L+2\right)}{2\left(\beta_x+L\alpha_x+1\right)^2}.
% \end{equation}
% With \eqref{pd1} and \eqref{pd2}, we can demonstrate that $\mathbb{E}[U_x]$ increases as $x$ increases. 
To proceed, we present the following theorem.
\begin{theorem}\label{TUincrease}
    % If $W_1>L+1$, $\mathbb{E}[U_x]$ increases as $x$ increases.
    If $W_x>L+1$, $\mathbb{E}[U_x]$ increases as $x$ increases; otherwise, $\mathbb{E}[U_x]=1$.
\end{theorem}
\begin{proof}
    See Appendix \ref{appC2}.
\end{proof}
% Since the smallest possible value of $\mathbb{E}[U_x]$ is $1$, Theorem \ref{TUincrease} indicates that $\mathbb{E}[U_x]$ is non-decreasing with respect to $x$.
Given that the smallest possible value of $\mathbb{E}[U_x]$ is $1$, Theorem \ref{TUincrease} indicates that $\mathbb{E}[U_x]$ is a non-decreasing function of $x$.
This suggests that, when the backoff contention window size is larger than $L+1$, an STA's expected time for its OBO counter to reset to zero lengthens after each failed transmission until the number of consecutive failures reaches $m$. 
% According to \eqref{AoIevolve}, with a sufficiently large value of $m$, if an STA experiences consecutive transmission failures, the expected instantaneous AoI of the STA will be significantly high just before it drops. 
% Consequently, we posit that a large $m$ is detrimental to the AoI performance. This assumption is supported by numerical results presented in Section~\ref{sec:simulation}. 
According to \eqref{AoIevolve}, with a sufficiently large value of $m$, if an STA experiences consecutive transmission failures with the backoff contention window size larger than $L+1$, the expected instantaneous AoI of the STA will be significantly high just before it drops. 
Consequently, we posit that a large $m$ is not beneficial to the AoI performance. This assumption is supported by numerical results presented in Section~\ref{sec:simulation}.
Consequently, to improve AoI performance in the network, we recommend setting a fixed OCW, specifically $m=0$.
% Although setting $m=0$ may increase the transmission collision rate, this issue can be alleviated by appropriately designing $W_0$.
We note that a similar method has been adopted in \cite{10778312}, where the backoff contention window size of the IEEE 802.11 distributed coordination function (DCF) mechanism is fixed and optimized to replicate an AoI-optimal Bernoulli transmission policy.
While setting $m = 0$ may increase the transmission collision rate and potentially degrade AoI performance, this issue can be addressed through the careful design of $W_0$.
With fixed $N$ and $L$, our goal is then to determine the optimal $\tilde{W}_0^*$ that minimizes $\overline{\Delta}$ when $m=0$. 
To that end, we first conduct some further theoretical analysis elaborated in the following subsection. 

\subsection{Theoretical Analysis}
% To avoid the problem that $q$ and $\rho$ have no closed-form solutions,
%To obtain more design insight for the proposed design concept,
We obtain the following theorem on AAoI $\overline{\Delta}$ when $m=0$.
\begin{theorem}\label{LB}
    Given $N$, $L$, $W_0$, with $\lambda=1$ and $m=0$, $\overline{\Delta}$ can be expressed as
    % \begin{equation}\label{AoISSo}
    % \begin{split}
    %     \overline{\Delta}&=\frac{\frac{1}{6}\alpha_0(\alpha_0+1)(2\alpha_0+1)L+(\alpha_0+1)^2\beta_0+1}{(\alpha_0+1)\alpha_0+2(\alpha_0+1)\beta_0+2}\\
    %     &+\left(\left(1-\frac{1}{L
    %     \mathbb{E}[U_0]}\right)^{1-N}-1\right)
    %     \mathbb{E}[U_0]+\frac{1}{2}, 
    % \end{split}
    % \end{equation}
    \begin{subequations}
    \begin{align}
        \overline{\Delta}&=\frac{\mathbb{E}[U_0^2]}{2\mathbb{E}[U_0]}+\frac{1-q}{q}\mathbb{E}[U_0]+\frac{1}{2}\label{UeAoISSo}\\
        &=\frac{\frac{1}{6}\alpha_0(\alpha_0+1)(2\alpha_0+1)L+(\alpha_0+1)^2\beta_0+1}{(\alpha_0+1)\alpha_0+2(\alpha_0+1)\beta_0+2}\label{AoISSo}\\
        &+\left(\left(1-\frac{1}{L
        \mathbb{E}[U_0]}\right)^{1-N}-1\right)
        \mathbb{E}[U_0]+\frac{1}{2}, \nonumber
    \end{align}
    \end{subequations}
    which is lower bounded by
    % \begin{equation}\label{AoILBSSo}
    % \begin{split}
    %     \overline{\Delta}_{LB}&=\left(\left(1-\frac{1}{L\mathbb{E}[U_0]}\right)^{1-N}-\frac{1}{2}\right)
    %     \mathbb{E}[U_0]+\frac{1}{2}.
    % \end{split}
    % \end{equation}
    \begin{subequations}
    \begin{align}
        \overline{\Delta}_{LB}&=\frac{\mathbb{E}[U_0]^2}{2\mathbb{E}[U_0]}+\frac{1-q}{q}\mathbb{E}[U_0]+\frac{1}{2}\label{UeAoILBSSo}\\
        &=\left(\left(1-\frac{1}{L\mathbb{E}[U_0]}\right)^{1-N}-\frac{1}{2}\right)
        \mathbb{E}[U_0]+\frac{1}{2}.\label{AoILBSSo}
    \end{align}
    \end{subequations}
    % In addition, $\overline{\Delta}$ and $\overline{\Delta}_{LB}$ are positively correlated.
\end{theorem}
\begin{proof}
See 
% Appendix C of \cite{CitationLiu}.
Appendix \ref{appB}.
\end{proof}

The expression of $\overline{\Delta}$ given in Theorem \ref{LB} is a piecewise multivariate multi-order function of $\alpha_0$ and $\beta_0$, which is significantly challenging to investigate thoroughly.
% Since $\overline{\Delta}_{LB}$ is more analytically feasible than $\overline{\Delta}$, we are motivated to minimize this lower bound of $\overline{\Delta}$ to guide the UORA parameter optimization.
To proceed, we consider its lower bound $\overline{\Delta}_{LB}$.
Clearly, the only difference between the expressions of $\overline{\Delta}$ and $\overline{\Delta}_{LB}$, given in \eqref{UeAoISSo} and \eqref{UeAoILBSSo}, is the numerator of the first term.
For \eqref{UeAoISSo}, the numerator of the first term is $\mathbb{E}[U_0^2]$, and for \eqref{UeAoILBSSo}, the numerator is $\mathbb{E}[U_0]^2$.
Our numerical results show that the ratio $(\mathbb{E}[U_0])^2/\mathbb{E}[U_0^2]$ is greater than $73\%$ under all feasible UORA configurations and this ratio increases as $L$ increases.
We thus conjecture that the deviation between $\overline{\Delta}$ and $\overline{\Delta}_{LB}$, defined as $(\overline{\Delta}-\overline{\Delta}_{LB})/\overline{\Delta}$, could be insignificant. 
We have validated this conjecture through numerical results, which show that the deviation remains below $5\%$ when $N$ and $L$ go large. This implies that $\overline{\Delta}$ and $\overline{\Delta}_{LB}$ have similar behaviors. 
Built upon this observation, we are motivated to minimize the lower bound $\overline{\Delta}_{LB}$ to guide the UORA parameter optimization.
% The complexities of the structures of \eqref{AoISSo} and \eqref{AoILBSSo} render them challenging to comprehend, thereby limiting the theoretical insights they offer.
Then, we find that the complexity of the expression structure of $\overline{\Delta}_{LB}$ renders it intractable to comprehend, thereby limiting the theoretical insights it offers.
% Hence, we consider $\overline{\Delta}_{LB}$ in a special case when $W_0-1$ is divisible by $L$. 
Hence, we consider $\overline{\Delta}_{LB}$ in a special case when $W_0-1$ is divisible by $L$, i.e., $W_0-1=L\alpha_0$.
Then, \eqref{AoILBSSo} can be rewritten and approximated as 
\begin{equation}\label{AoIapproLBSS}
\begin{split}
    &\left((1-\frac{1}{U(W_0)L})^{1-N}-\frac{1}{2}\right)U(W_0)+\frac{1}{2}\triangleq\hat{\Delta}_{LB}\\
    &\overset{(a)}{\approx}\left(\exp\left(\frac{N-1}{U(W_0)L}\right)-\frac{1}{2}\right)U(W_0)+\frac{1}{2}\triangleq \tilde{\Delta}_{LB},
\end{split}
\end{equation}
in which function
% $U(W_0)=\frac{W_0^2+(L-2)W_0+L+1}{2W_0L}$,
\begin{equation}\label{UFSS}
    U(W)=\frac{W^2+(L-2)W+L+1}{2WL},
\end{equation}
where $(a)$ is obtained by applying the approximation $(1-x)^n\approx \exp(-nx)$ for $x\in(0,1)$ and large $n\in \mathbb{N}$ to the term $q=(1-\frac{\rho}{L})^{N-1}$.
Compared with \eqref{AoILBSSo}, $\tilde{\Delta}_{LB}$ has fewer independent variables and degrees, and thus is more tractable for analysis.

Define the function 
% $\tilde{\Delta}_{LB}(W)=\left(\exp\left(\frac{N-1}{U(W)L}\right)-\frac{1}{2}\right)U(W)+\frac{1}{2}$.
\begin{equation}
    \tilde{\Delta}_{LB}(W)=\left(\exp\left(\frac{N-1}{U(W)L}\right)-\frac{1}{2}\right)U(W)+\frac{1}{2}.
\end{equation}
Then, we have the following theorem:
\begin{theorem}\label{LBroot}
    Given the integers $N>1$ and $L>1$, when
    \begin{subequations}
        \begin{align}
            &B^2-4(L+1)>0,\label{C1}\\
            &B<0\label{C2}
        \end{align}
    \end{subequations}
    hold, function $\tilde{\Delta}_{LB}(W)$ has three positive real roots, presented as
    \begin{subequations}
    \begin{align}
        &r_1=\frac{-B-\sqrt{B^2-4(L+1)}}{2},\label{R1}\\
        &r_2=\sqrt{L+1},\label{R2}\\
        &r_3=\frac{-B+\sqrt{B^2-4(L+1)}}{2},\label{R3}
    \end{align}  
    \end{subequations}
    where 
    % $B=-\frac{2(N-1)}{\mathcal{W}_0(-\frac{1}{2e})+1}+L-2$,
    \begin{equation}
        B=-\frac{2(N-1)}{\mathcal{W}_0(-\frac{1}{2e})+1}+L-2,
    \end{equation}
    and $\mathcal{W}_0(x)$ is the $0$ branch of the Lambert W function $\mathcal{W}(x)$ \cite{corless1996lambert}.
    More specifically, $r_1<r_2<r_3$, and $\tilde{\Delta}_{LB}(W)$ decreases monotonically in $(0,r_1)$ and $[r_2,r_3)$ while increases monotonically in $[r_1,r_2)$ and $[r_3,+\infty)$.
    On the contrary, when condition \eqref{C1} and \eqref{C2} do not hold at the same time, $\tilde{\Delta}_{LB}(W)$ has only one real positive root $r_2=\sqrt{L+1}$.
    Further, $\tilde{\Delta}_{LB}(W)$ decreases and increases in $(0,r_2)$ and $[r_2,+\infty)$, respectively.
    
\end{theorem}
\begin{proof}
See 
% Appendix D of \cite{CitationLiu}.
Appendix \ref{appC}.
\end{proof}

Theorem \ref{LBroot} delineates two possible behaviors of $\tilde{\Delta}_{LB}(W)$ as $W$ increases from $0$ to $+\infty$: (1) $\tilde{\Delta}_{LB}(W)$ initially decreases and then increases, and 
(2) $\tilde{\Delta}_{LB}(W)$ undergoes an initial decrease, followed by an increase, a subsequent decrease, and an eventual increase.
These observations are key to developing an efficient AoI-oriented optimization algorithm for UORA networks, introduced in the next subsection.

\subsection{Optimization Algorithm Design}\label{AD}
Recall that we aim to minimize the lower bound of $\overline{\Delta}$ to guide the UORA parameter optimization.
This can be achieved by using the behavior of $\overline{\Delta}_{LB}(W)$, defined as \eqref{AoILBSSo} by substituting $\alpha_0=\lfloor(W-1)/L\rfloor$ and $\beta_0=W-1-\alpha_0L$.
Note that we only obtain the behavior of an approximation of $\overline{\Delta}_{LB}(W)$ in $(0,+\infty)$ when $W - 1$ is divisible by $L$, namely $\tilde{\Delta}_{LB}(W)$. However, it is not immediately evident that $\tilde{\Delta}_{LB}(W)$ exhibits similar behavior to $\overline{\Delta}_{LB}(W)$. 
To proceed, we define function $\hat{\Delta}_{LB}(W)$ as the left hand side of \eqref{AoIapproLBSS} by replacing $W_0$ with $W$ and define 
\begin{equation}
    \boldsymbol{W}_{0,z}^M\triangleq \{2^z|z\in\mathbb{N}\} \bigcap \{zL+1|z\in\mathbb{N}\}.
\end{equation}
Clearly, $[\overline{\Delta}_{LB}(W_0)]_{W_0\in \boldsymbol{W}_{0,z}^M}$ is the collection of $\overline{\Delta}_{LB}$ when $W_0-1$ is divisible by $L$.
% Clearly, elements of $\boldsymbol{W}_{0,z}^M$ are admissible for both $\overline{\Delta}_{LB}(W)$ and $\hat{\Delta}_{LB}(W)$.
We thus have $[\hat{\Delta}_{LB}(W_0)]_{W_0\in \boldsymbol{W}_{0,z}^M}=[\overline{\Delta}_{LB}(W_0)]_{W_0\in \boldsymbol{W}_{0,z}^M}$, implying that $[\hat{\Delta}_{LB}(W_0)]_{W_0\in \boldsymbol{W}_{0,z}^M}$ can be fitted to both $\hat{\Delta}_{LB}(W)$ and $\overline{\Delta}_{LB}(W)$ in $(0,+\infty)$.
Thus, $\overline{\Delta}_{LB}(W)$ and $\hat{\Delta}_{LB}(W)$ tend to have similar behaviors in $(0,+\infty)$.
Considering $\tilde{\Delta}_{LB}(W)$ is approximate to $\hat{\Delta}_{LB}(W)$, we conjecture that $\tilde{\Delta}_{LB}(W)$ also tends to have the similar behavior in $(0,+\infty)$. We will validate this conjecture using numerical results in Section \ref{sec:simulation}.

With $m=0$, we tend to consider admissible $W_0$ close to $W^*=\min_W\tilde{\Delta}_{LB}(W)$, i.e., the positive roots of $\tilde{\Delta}_{LB}(W)$, to be the AoI-oriented UORA parameter setups.
Recalling Corollary \ref{AoINoChange} and $r_1<r_2=\sqrt{L+1}<L+1$, we only need to consider one of $r_1$ and $r_2$.
Following Theorem \ref{LBroot}, we develop an efficient parameter search algorithm, formally described in  \textbf{Algorithm \ref{algLC}}.
% where $\emptyset$ denotes the empty set.
% \begin{algorithm}
%     \caption{An efficient search algorithm}\label{algLC}
%     \textbf{Initialization}: $N$, $L$, $m\gets 0$, and $\boldsymbol{\mathcal{E}}\gets\emptyset$\;
%     \eIf{$B<0$ \&\& $B^2-4(L+1)>0$}{
%     Calculate $r_1$ and $r_3$ by \eqref{R1} and \eqref{R3}, respectively\;
%     $E_1\gets \min\{\max\{\log_2(r_1),1\},1023\}$\; $E_3\gets\min\{\max\{\log_2(r_3),0\},1023\}$\;
%     \For{$i\gets 0,3$}{
%         \eIf{$E_i\in \mathbb{N}$}{
%         $\boldsymbol{\mathcal{E}}\gets\boldsymbol{\mathcal{E}}\bigcup\{E_i\}$\;
%         }
%         {$\boldsymbol{\mathcal{E}}\gets\boldsymbol{\mathcal{E}}\bigcup\{\lfloor E_i\rfloor,\lceil E_i\rceil\}$\;}
%     }
%   }{Calculate $r_2$ by \eqref{R2}\;
%   $E_2\gets \min\{\max\{\log_2(r_2),0\},1023\}$\;
%   $i\gets 2$ and obtain $\boldsymbol{\mathcal{E}}$ following steps $7$ to $10$\;
%   }
%   \textbf{Output}: 
%   $\tilde{W}_0^*\gets \arg\min_{W_0=2^e,e\in\boldsymbol{\mathcal{E}}}\overline{\Delta}$ and $m$\;
% \end{algorithm}
\begin{algorithm}
    \caption{An efficient parameter search algorithm}\label{algLC}
    \textbf{Initialization}: $N$, $L$, $m\gets 0$\;
    \eIf{$B<0$ \&\& $B^2-4(L+1)>0$}{
    Calculate $r_3$ by \eqref{R3}\;
    % $E_1\gets \min\{\max\{\log_2(r_1),1\},1023\}$\;
    $E\gets\min\{\max\{\log_2(r_3),L+1\},7\}$\;
    % \For{$i\gets 0,3$}{
        \eIf{$E\in \mathbb{N}$}{
        $\boldsymbol{\mathcal{E}}\gets\{E\}$\;
        }
        {$\boldsymbol{\mathcal{E}}\gets\{\lfloor E\rfloor,\lceil E\rceil\}$\;}
    % }
  }{Calculate $r_2$ by \eqref{R2}\;
  $E\gets \min\{\log_2(r_2),7\}$\;
  % $E\gets \min\{\max\{\log_2(r_2),0\},1023\}$\;
  % $i\gets 2$ and 
  Obtain $\boldsymbol{\mathcal{E}}$ following steps $5$ to $9$\;
  }
  \textbf{Output}: 
  $\tilde{W}_0^*\gets \arg\min_{W_0=2^e,e\in\boldsymbol{\mathcal{E}}}\overline{\Delta}$ and $m$\;
\end{algorithm}
Compared with the exhaustive search method, \textbf{Algorithm \ref{algLC}} only needs to compare the AAoI of three possible values of $W_0$, resulting in significantly higher efficiency.

%However, Theorem \ref{LBroot} cannot be applied in the UORA network with the stochastic arrival model, i.e., $\lambda<1$, rendering Algorithm \ref{algLC}
% the improvement of $\overline{\Delta}$ through the optimization of its lower bound 
%ineffective.
Finally, we discuss how to extend Algorithm \ref{algLC} to a more general case with $\lambda<1$. Specifically, we devise a search policy focused on identifying $W_0$ that minimizes $\overline{\Delta}$ with $m=0$.
% , given in Algorithm \ref{algLCv2}. 
% To proceed , by Corollary \ref{AoINoChange} and Theorem \ref{LBroot}, we intuitively assume that $\overline{\Delta}$ with $EOCW_{min}\ge \log_2(L+1)$ behaves similarly to $\tilde{\Delta}_{LB}(W)$ in $(r_2,+\infty)$, and only has one local minimum point when $m=0$ and $\lambda<1$.
{
Our conjecture is that given $m=0$, the AAoI $\overline{\Delta}$ of UORA networks with stochastic arrival of status updates and that with generate-at-will status updates tend to have similar behaviors.
The rationale is that after $\mathcal{M}_1$ of a UORA network with $\lambda<1$ enters its steady state, the expected number of active STAs in each slot is $\sum^N_{a=1}a\mu_a$.
Building on this, we can consider the UORA network with $\lambda < 1$ as a generate-at-will UORA network comprising $\sum^N_{a=1}a\mu_a$ STAs. 
The effectiveness of this conjecture has been confirmed through numerical results; however, establishing a rigorous theoretical foundation for its application remains an area for future research. 
Under this conjecture and Theorem \ref{LB}, we state that $\overline{\Delta}$ behaves similarly to $\tilde{\Delta}_{LB}(W)$.
Moreover, by Corollary \ref{AoINoChange} and Theorem \ref{LBroot}, we only need to consider $\tilde{\Delta}_{LB}(W)$ in $(r_2,+\infty)$ for the parameter optimization, which has at most one local minimum point.}
Built upon this, we propose an extended efficient search algorithm, given in Algorithm~\ref{algLCv2}.
\begin{algorithm}
    \caption{An extended efficient search algorithm}\label{algLCv2}
    \textbf{Initialization}: $N$, $L$, $\lambda$, $m\gets 0$, $e\gets \lfloor\log_2(L+1)\rfloor$\;
    Compute $\overline{\Delta}$ with $(N,L,\lambda,m,W_0=2^e)$ by \eqref{AoIUE}\;
    $D_0\gets \overline{\Delta}$\;
    \While{$e\le 7$}{
    $e\gets e+1$\;
    Compute $\overline{\Delta}$ with $(N,L,\lambda,m,W_0=2^e)$ by \eqref{AoIUE}\;
    $D_1\gets \overline{\Delta}$\;
    \If{$D_1>D_0$}
    {
        $\tilde{W}_{0,e}^*\gets W_0/2$\;
        break\;
    }
    $D_0\gets D_1$\;
    }
  \textbf{Output}: 
  $\tilde{W}_{0,e}^*$ and $m$\;
\end{algorithm}
\vspace{-1em}
% This algorithm is empirically formulated based on the assumption that the behavior of $\overline{\Delta}$ is likely to be similar to that of $\hat{\Delta}_{LB}(W)$ when $\lambda<1$.
 %The theoretical foundation of this assumption remains an area for future investigation.

\section{Simulation Results}\label{sec:simulation}
In this section, we first demonstrate the accuracy of the proposed analytical model.
% Subsequently, we verify the theoretical analyses of the AoI performance in the generate-at-will model with $m=0$, followed by evaluating the effectiveness of the proposed algorithm through comparison with the exhaustive search method and the round-robin scheme. 
Subsequently, we verify the theoretical analyses of the AoI performance in the generate-at-will model with $m=0$, followed by evaluating the effectiveness of the proposed algorithm through comparison with the exhaustive search method and two centralized schemes.
In the following figures, each simulation result is generated by averaging over $10^6$ Monte-Carlo simulation~runs.
Note that we present the relevant parameters within the discussion or alongside the corresponding figures for each simulation.

% \begin{figure}
%      \centering
%         % \captionsetup{labelformat=empty}
%      \begin{subfigure}[t]{0.45\textwidth}
%          \centering
%          \includegraphics[width=\textwidth]{q_vs_La_Emin=2_m=4.eps}
%          \caption{Successful transmission probability.}
%          \label{Aq}
%      \end{subfigure}
%      % \hfill
%      \begin{subfigure}[t]{0.45\textwidth}
%          \centering
%          \includegraphics[width=\textwidth]{rho_vs_La_Emin=2_m=4.eps}
%          \caption{Accessing probability of an active STA.}
%          \label{Arho}
%      \end{subfigure}
%      \hfill
%         \caption{UORA accessing performances versus the packet arrival rate, $\lambda$, with $EOCW_{min}=2$ and $m=4$.}
%         \label{q_rho_v_La}
%     \vspace{-1em}
% \end{figure}
\begin{figure}[t]
	\centering    \includegraphics[width=0.45\textwidth]{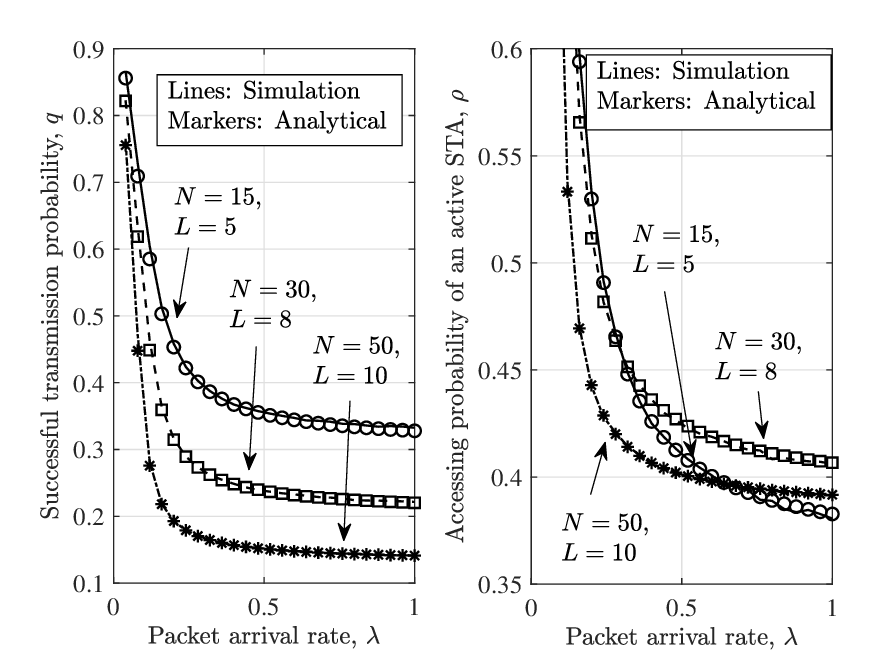}
	\caption{Successful transmission rate $q$ and conditional accessing rate $\rho$ versus packet arrival rate $\lambda$.}
	\label{q_rho_v_La}
\vspace{-1em}
\end{figure}
Fig. \ref{q_rho_v_La} includes two subplots demonstrating the evaluations of the successful transmission rate, $q$, and the probability that an active STA accesses one of the RUs, $\rho$, in different setups, respectively.
The configurations of $N$ and $L$ are given in the subplots.
This evaluation can assess the effectiveness of the analytical model based on $\mathcal{M}_1$ and $\mathcal{M}_2$.
We set $EOCW_{min}=2$ and $m=4$.
We can observe from Fig. \ref{q_rho_v_La} that the simulation and analytical results align closely in each configuration, verifying the effectiveness of our analytical model.
% Additionally, Fig. \ref{Aq} shows that $q$ decreases as $\lambda$ increases in all scenarios.
% Similarly, as $\lambda$ increases, the curves of $\rho$ consistently decrease in Fig. \ref{Arho}. 
Additionally, the left subplot shows that $q$ decreases as $\lambda$ increases in all scenarios.
Similarly, as $\lambda$ increases, the curves of $\rho$ consistently decrease in the right subplot. 
These behaviors are intuitive because a higher packet arrival rate increases the number of the active STAs as well as the accessing STAs, resulting in more frequent transmission collisions.
On the other hand, due to the backoff mechanism, the growth rate of the number of the accessing STAs is slower than that of the active STAs, which leads to a reduction in the accessing rate, $\rho$, of the active STAs.

\begin{figure}[t]
	\centering    \includegraphics[width=0.45\textwidth]{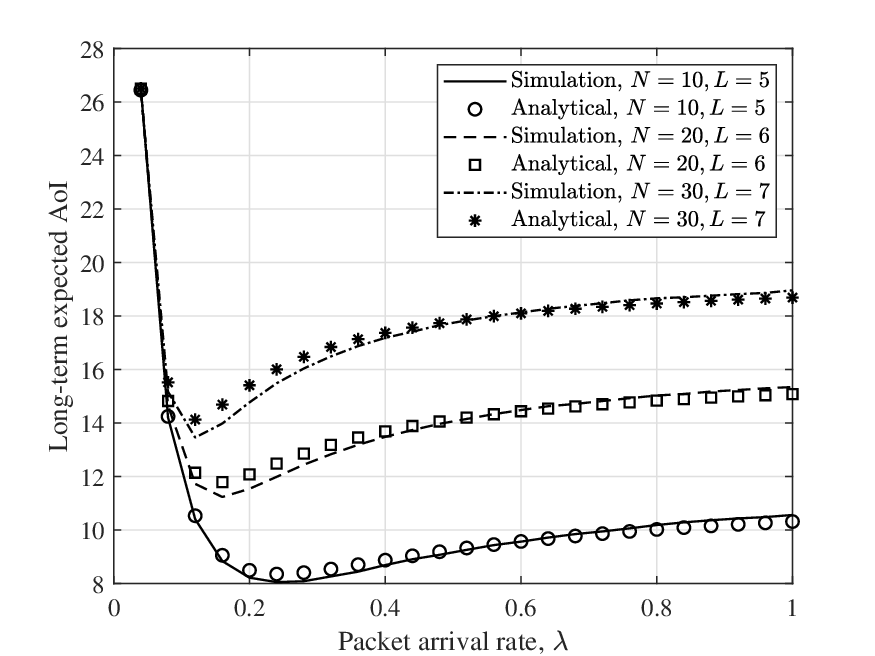}
	% \caption{AAoI versus packet arrival rate $\lambda$.}
    \caption{AAoI versus packet arrival rate $\lambda$, with $EOCW_{min}=3$ and $m=3$.}
	\label{AoI_v_La}
\vspace{-1em}
\end{figure}
Fig. \ref{AoI_v_La} depicts the AAoI performance of the UORA network versus the packet arrival rate in different setups.
% We can see from Fig. \ref{AoI_v_La} that the discrepancies between the analytical and simulation outcomes are less than 0.5\%, demonstrating the validity of the proposed analytical model.
We can see from Fig. \ref{AoI_v_La} that the discrepancies between the analytical and simulation results are larger compared to Fig. \ref{q_rho_v_La}.
This is due to the adopted approximation for analyzing the expected service time $\mathbb{E}[S]$, which may become less accurate for certain values of $\lambda$.
Nevertheless, the discrepancies are less than 0.5\%, confirming the validity of the proposed analytical model.
Furthermore, Fig. \ref{AoI_v_La} shows that the AAoI first decreases and subsequently increases as $\lambda$ increases.
This is intuitive as the increase of $\lambda$ increases the number of transmissions in the network, reducing system times while increasing the probability of transmission collisions.
When $\lambda$ is slightly large, the effect of the reduction of system times is greater than that of the increased probability of collision, resulting in the decrease of the curves.
% Moreover, a significantly high $\lambda$ introduces a large number of transmissions, increasing the probability of transmission collisions, which leads to the EAoIs ultimately increasing.
% On the contrary, the impact of increased probability of collision dominates with a significantly large $\lambda$, which leads to the EAoIs ultimately increasing.
In contrast, the effect of increased collision probability dominates when $\lambda$ increases significantly, leading to an eventual increase in~AAoI.
% \begin{figure}[t]
% 	\centering    \includegraphics[width=0.45\textwidth]{AoI_vs_Emin_La=0.6_N=15_L=5_vB.eps}
% 	\caption{AAoI versus $EOCW_{min}$.}
% 	\label{AoI_v_Emin}
%  \vspace{-1em}
% \end{figure}
\begin{figure}[t]
	\centering    \includegraphics[width=0.45\textwidth]{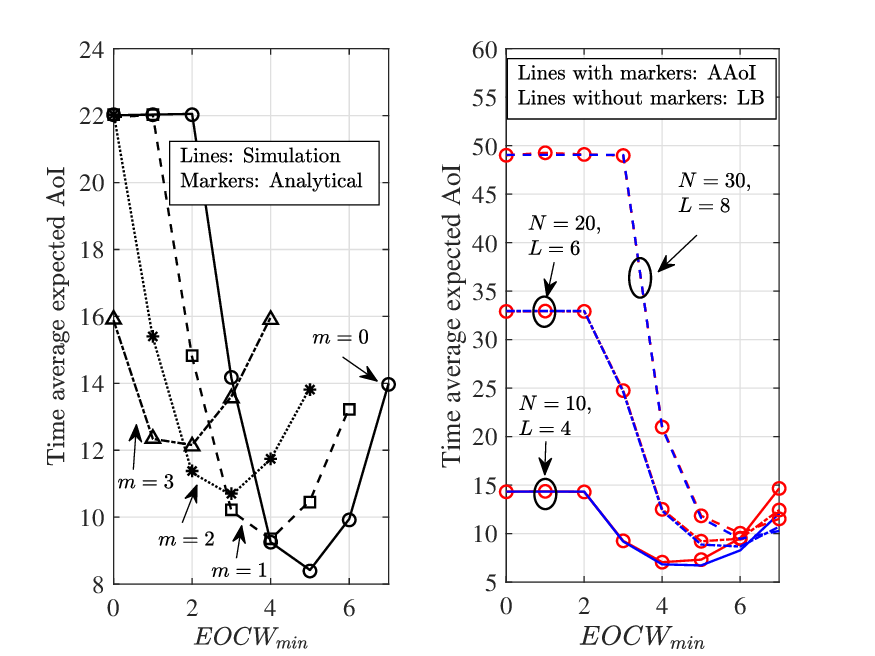}
	\caption{AAoI versus $EOCW_{min}$.}
	\label{AoI_v_Emin}
 \vspace{-1em}
\end{figure}

% In Fig. \ref{AoI_v_Emin}, we show the AAoI performance of the UORA network versus $EOCW_{min}$ with different configurations.
% We set $N=15$, $L=5$, and $\lambda=0.6$.
% Fig. \ref{AoI_v_Emin} demonstrates that, when $m=0$, AAoI remains unchanged as $EOCW_{min}$ increases from $0$ to $2$. 
% This can be explained by Corollary \ref{AoINoChange} that $\overline{\Delta}$ is the same when $W_m\le L+1$. 
% Fig. \ref{AoI_v_Emin} also shows that the AAoI is initially non-increasing and then non-decreasing as $EOCW_{min}$ increases when $m=0,2$.
% This is understandable as the increase of $EOCW_{min}$ causes $W_0$ to rise, thereby reducing the number of simultaneous transmissions, which in turn increases the transmission success rate.
% However, by \eqref{pd1} and \eqref{pd2}, the initialized OBO counter of an STA may be considerably large as $W_0$ increases, leading to the curves eventually increasing.
% Additionally, the AAoI exhibits an initial increase, followed by a decrease, and eventually rises again as $EOCW_{min}$ increases for $m=4,6$.
% % These behaviors are similar to that of samples of $\tilde{\Delta}_{LB}(W)$ described in Theorem \ref{LBroot}.
% This observation is intuitive because a relatively large $m$ increases the probability of initializing a large OBO counter after a transmission failure, and thus increases the AAoI even if $EOCW_{min}$ is small.
% %This observation is due to that a relatively large $m$ increases the probability a large OBO counter is initialized after a failed transmission even when $EOCW_{min}$ is small.
In the left subplot of Fig. \ref{AoI_v_Emin}, we show the AAoI performance of the UORA network versus $EOCW_{min}$ with $m=0,1,2,3$, respectively.
We set $N=15$, $L=5$, and $\lambda=0.6$.
Note that the curves for $m=1,2,3$ are not plotted completely in the left subplot since $EOCW_{max}=EOCW_{min}+m\le 7$.
% $EOCW_{max}$'s of the corresponding configurations are not in the predefined IEEE 802.11ax $EOCW$ range.
The left subplot demonstrates that when $m=0,1$, AAoI initially remains unchanged as $EOCW_{min}$ increases. 
This can be explained by Corollary \ref{AoINoChange} that $\overline{\Delta}$ is the same when $W_m\le L+1$. 
The left subplot also shows that the AAoI is initially non-increasing and then non-decreasing as $EOCW_{min}$ increases for all settings.
This is understandable as the increase of $EOCW_{min}$ causes $W_0$ to rise, thereby reducing the number of simultaneous transmissions, which in turn increases the transmission success rate.
However, by \eqref{pd1} and \eqref{pd2}, the initialized OBO counter of an STA may be considerably large as $W_0$ increases, 
% leading to the curves eventually increasing.
leading to a significant decrease in the number of accessing STAs, which yields the eventual increases of the curves.

% \begin{figure}[t]
% 	\centering    \includegraphics[width=0.45\textwidth]{AoI_vs_N_La=0.4_Emin=2_m=4.eps}
% 	\caption{AAoI versus the number of the STAs of $N$, where $\lambda=0.4$, $EOCW_{min}=2$, and $m=4$.}
%     \Description[AAoI versus $N$.]{AAoI versus the number of the STAs of $N$, where $\lambda=0.4$, $EOCW_{min}=2$, and $m=4$.}
% 	\label{AoI_v_N}
% \end{figure}

% \begin{figure}[t]
% 	\centering    \includegraphics[width=0.45\textwidth]{AoIandLB_vs_Emin_La=1_m=0_vB.eps}
% 	\caption{AAoI and its lower bound $\overline{\Delta}_{LB}$ versus $EOCW_{min}$ with $\lambda=1$ and $m=0$.}
% 	\label{AoI&LB_v_Emin}
%  \vspace{-1em}
% \end{figure}
The right subplot of Fig. \ref{AoI_v_Emin} plots the AAoI and their proposed lower bounds with the increasing $EOCW_{min}$ and $m=0$ under the generate-at-will model.
We set three parameter configurations as follows: $N=10, L=4$; $N=20, L=6$; and $N=30, L=8$.
We can observe from the right subplot that the lower bound is consistently smaller than the associated AAoI performance across all configurations, validating our analysis presented in Theorem \ref{LB}.
The right subplot also shows that 
% the curve of the lower bound is similar to that of the AAoI in each configuration.
the AAoI and its lower bound have similar trends.
% the AAoI and its lower bound have similar trends because of the positive correlation between them.

\begin{figure}[t]
	\centering    \includegraphics[width=0.44\textwidth]{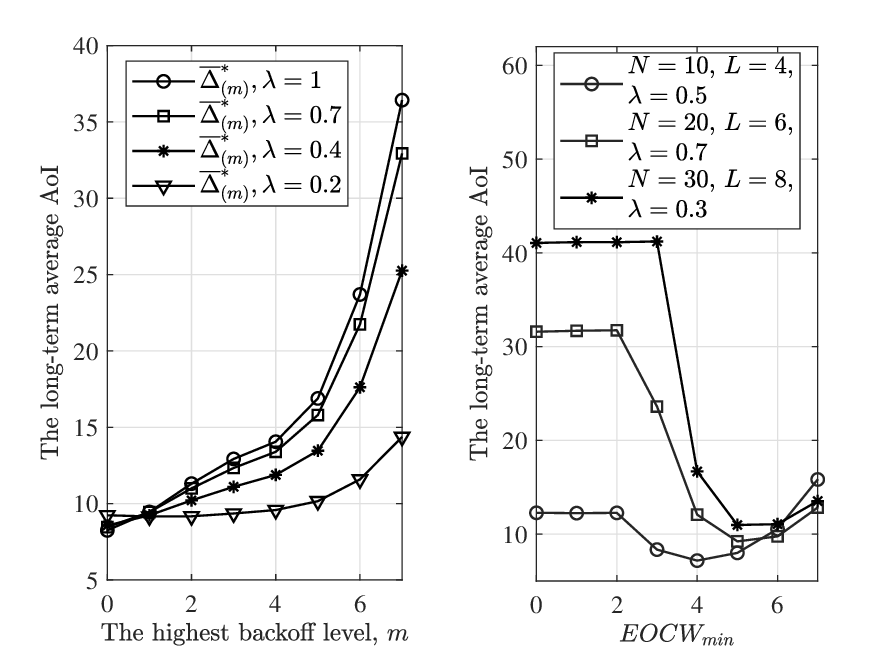}
	\caption{Verifications of the conjectures in algorithm design.}
	\label{AV}
 \vspace{-1em}
\end{figure}
% Fig. \ref{AV} offers two subplots used to verify the conjectures adopted in the design of parameter searching algorithms.
% The left subplot shows the AAoI versus the highest backoff level $m$ across different $\lambda$.
Fig. \ref{AV} offers two subplots used to verify the conjectures adopted in the design of parameter searching algorithms.
The left subplot shows the AAoI versus the highest backoff level $m$, across varying values of $\lambda$ being $1,0.7,0,4$, and $0.2$.
We set $N=12$ and $L=4$.
% Specifically, the special exhaustive search method outputs $W_{0,m}^*=\arg\min_{W_0}\overline{\Delta}$ given a fixed $m$.
Specifically, for a fixed $m$, we exhaustively search over $W_0$ to find $W_{0,m}^*=\arg\min_{W_0}\overline{\Delta}$. 
We denote the AAoI under parameters $(W_{0,m}^*,m)$ by $\overline{\Delta}_{(m)}^*$. 
Fig. \ref{AV} shows that $\overline{\Delta}_{(m)}^*$ deteriorates as $m$ increases for different $\lambda$.
This observation provides an empirical rationale for our algorithm design focusing on $m=0$.
% The right subplot depicts curves of AAoI as $EOCW_{min}$ increases in different setups with $m=0$.
The right subplot depicts curves of AAoI as $EOCW_{min}$ increases with $m=0$ across three different sets of configurations being $N=10,L=4,\lambda=0.5$, $N=20,L=6,\lambda=0.7$, and $N=30,L=8,\lambda=0.3$.
We observe that the AAoI has one local minimum when $\lambda<1$, consistent with our assumption in Section~\ref{AD}.

\begin{figure}[t]
	\centering    \includegraphics[width=0.45\textwidth]{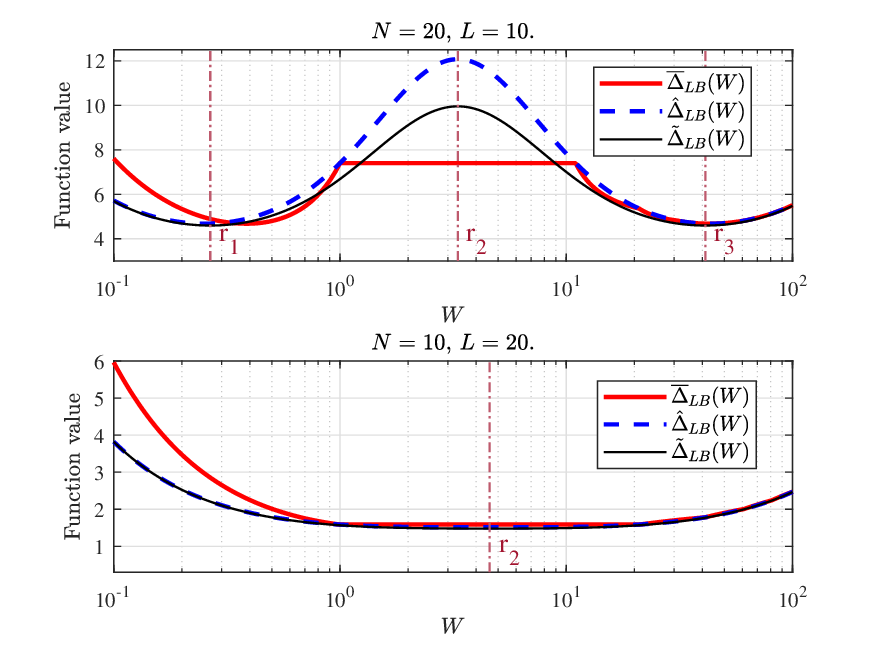}
	\caption{Function $\overline{\Delta}_{LB}(W)$, $\hat{\Delta}_{LB}(W)$ and $\tilde{\Delta}_{LB}(W)$.}
	\label{ThreeFunctions}
 \vspace{-1em}
\end{figure}
In Fig. \ref{ThreeFunctions}, we plot the curves for three functions: $\overline{\Delta}_{LB}(W)$, $\hat{\Delta}_{LB}(W)$ and $\tilde{\Delta}_{LB}(W)$ in the range  $W\in(0,+\infty)$, for two cases, shown in two subplots.
In the first case, we set $N=20$ and $L=10$, which satisfies conditions \eqref{C1} and \eqref{C2}.
In the second case, we set $N=10$ and $L=20$, making conditions \eqref{C1} and \eqref{C2} invalid.
The roots are calculated by \eqref{R1}, \eqref{R2}, and \eqref{R3}.
We can see from Fig. \ref{ThreeFunctions} that the trajectories of $\tilde{\Delta}_{LB}(W)$ in the two cases follow that given in Theorem \ref{LBroot}.
The locations of the roots in both cases align with Theorem \ref{LBroot} as well.
Moreover, $\overline{\Delta}_{LB}(W)$ remains unchanged within a segment of $W$ in each case.
The rationale is that these segments correspond to the cases described in Corollary \ref{AoINoChange}.
% Fig. \ref{ThreeFunctions} exhibits that the behaviors of $\overline{\Delta}_{LB}(W)$ and $\hat{\Delta}_{LB}(W)$ are analogous, in line with our assumption.

\begin{figure}[t]
	\centering    \includegraphics[width=0.44\textwidth]{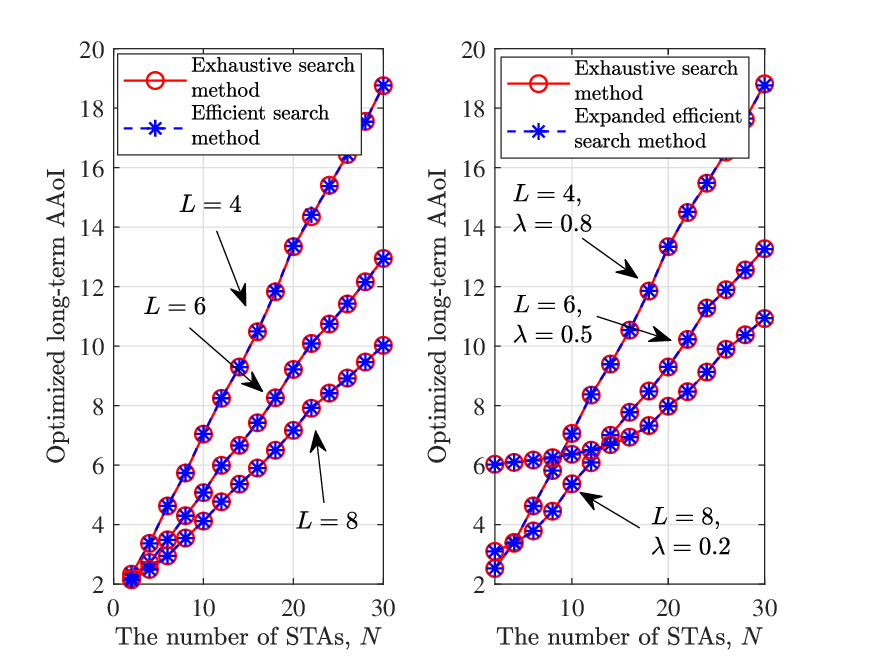}
	\caption{Optimized AAoI versus the number of the STAs, $N$, in two setups.}
	\label{FS_v_LC}
 \vspace{-1em}
\end{figure}
Fig. \ref{FS_v_LC} illustrates the optimized AAoI performances of the UORA networks versus $N$ in two different network setups.
Each setup features three parameter configurations, which are displayed in the respective subplots.
% In Fig. \ref{LC1}, the generate-at-will model, i.e., $\lambda=1$, is considered, and we compare the exhaustive search method with the efficient search method outlined in Algorithm \ref{algLC}. 
% Additionally, in Fig. \ref{LC2}, we consider status updates randomly arriving at the 
% STAs with the arrival rates shown near the corresponding curves. 
In the left subplot, the generate-at-will model, i.e., $\lambda=1$, is considered, and we compare the exhaustive search method with the efficient search method outlined in Algorithm \ref{algLC}. 
Additionally, in the right subplot, we consider status updates randomly arriving at the 
STAs with the arrival rates shown near the corresponding curves.
We also compare the extended efficient search presented in Algorithm \ref{algLCv2} and the exhaustive search approach. 
It is depicted in Fig. \ref{FS_v_LC} that the proposed efficient parameter search algorithms perform comparably to the corresponding exhaustive search methods, which minimize the AoI performance of UORA networks. 
The observation in the right subplot confirms that the proposed extended efficient search policy can also be applied in UORA networks with the stochastic arrival model.

% We finally compare the AAoI performance of UORA networks optimized by the proposed method with that of networks adopting a centralized OFDMA strategy.
We finally compare the AAoI performance of UORA networks optimized by the proposed method with that of networks adopting centralized OFDMA strategies.
It is important to note that centralized OFDMA suffers from additional overhead compared to UORA. Specifically, centralized OFDMA requires extra overhead to schedule STAs and to acquire local information of STAs, and it necessitates additional overhead for re-allocating indexes of the STAs when there is frequent entry and exit from the network. 
% To address this, we consider a centralized strategy with low overhead, named the round-robin (RR) policy, which schedules $L$ STAs in each slot in a circular order. 
% In the simulation, we assume that STAs remain within the network, a scenario more favorable to the RR policy since it avoids additional overhead.
To address this, we consider two centralized strategies with low overhead, named the round-robin (RR) policy, which schedules $L$ STAs in each slot in a circular order, and the max-AoI (MA) policy, which always schedules STAs with the largest $L$ AoI values.
In the simulation, we assume that STAs remain within the network, a scenario more favorable to the RR and MA policies since it avoids additional overhead.
The AAoI performances of the optimized UORA and the two centralized strategies against varying packet arrival rates $\lambda$ are simulated under two sets of parameters, $N=30$, $L=3$ and $N=100$, $L=5$.
The simulation results are shown in Fig. \ref{UORAvsRRvsMA}, indicating that UORA moderately outperforms RR when the status update arrival rate at STAs is low. The rationale is that the subchannel allocation in the RR policy is static, leading to subchannels being occupied even when their STAs have no packets to transmit. Consequently, active STAs under the RR policy can experience longer waiting times before transmissions compared to those in the UORA network.
We can also see from Fig. \ref{UORAvsRRvsMA}, with infrequent status update arrivals, UORA is superior to the MA policy.
This is because the AP employing the MA policy might waste considerable actions on repeatedly scheduling idle STAs with the largest AoI when the status updates arrive at the STA with a low~rate.
\begin{figure}[t]
	\centering    \includegraphics[width=0.45\textwidth]{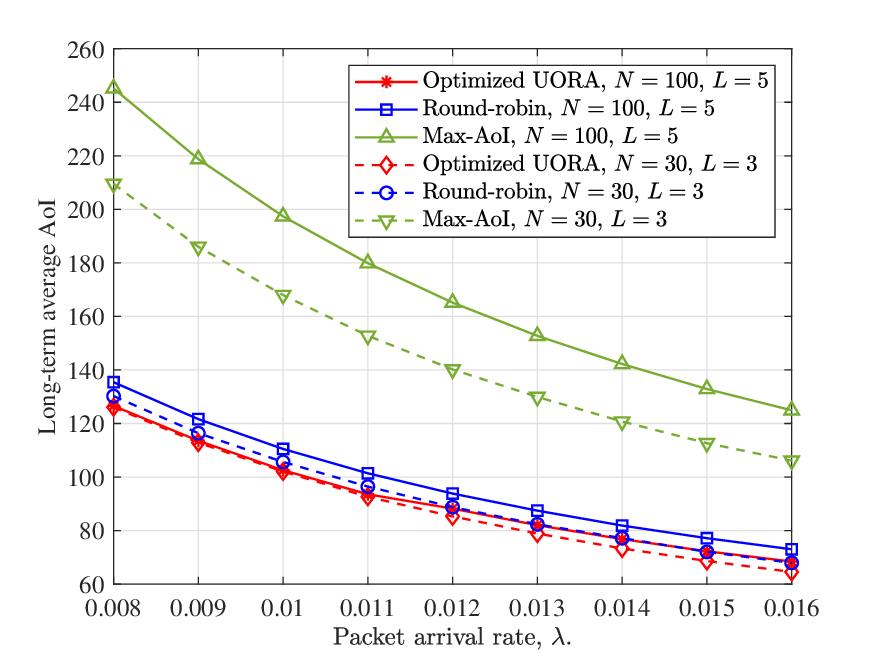}
	\caption{Comparisons of the AAoI performance between UORA, the round-robin policy, and the MA policy.}
	\label{UORAvsRRvsMA}
 \vspace{-1em}
\end{figure}

\section{Conclusion}
In this paper, we made the initial effort to analyze and optimize the age of information (AoI) performance of the uplink orthogonal frequency division multiple access-based random access (UORA) mechanism introduced in the latest WiFi standard.
By developing an analytical framework to characterize the transmission states of stations in UORA networks, we managed to evaluate their long-term average AoI (AAoI) using numerical methods. 
This allows us to optimize the AoI performance by tuning the UORA parameters through the exhaustive search.
We further derived closed-form expressions for the AAoI and its approximated lower bound with a fixed contention window and generate-at-will status updates.
In light of this, we proposed efficient parameter optimization algorithms, which can be implemented with only a few iterations of searching.
Simulation results corroborated our analyses and illustrated that the proposed algorithms offered the AAoI comparable to the optimal AoI performance achieved by exhaustive search. 
Moreover, the AAoI of the optimized UORA scheme outperforms that of the round-robin and max-AoI policies in large and low-traffic networks.
% \newpage
%\appendices
 \section*{APPENDIX}

\subsection{Proof of Theorem \ref{T0}}\label{appA0}
By Bayes theorem and the law of total probability, we have
    \begin{equation}\label{STY}
        q=\Pr\{ST|Y\}=\sum_{a=0}^{N-1}\Pr\{Y_a|Y\}\Pr\{ST|Y,Y_a\},
    \end{equation}
    where $Y$ indicates a certain STA is accessing, $Y_a$ represents that there are $a$ other active STAs in the BSS, $ST$ indicates the accessing STA has a successful transmission, and $\Pr\{Y_a|Y\}=\Pr\{Y_a,Y\}/\sum_{i=0}^{N-1}\Pr\{Y_i,Y\}$.
    Due to the symmetry of the network, the probability that certain $i'$ STAs are active is $\mu_{i'}/\binom{N}{i'}$, and thus
    \begin{equation}\label{YaY}
        \Pr\{Y_a,Y\}=\frac{\rho\mu_{a+1}\binom{N-1}{a}}{\binom{N}{a+1}}=\frac{\rho(a+1)\mu_{a+1}}{N}.
    \end{equation}
    Furthermore, $\Pr\{ST|Y,Y_a\}$ is the probability that other accessing STAs do not occupy the same RU as the certain STA, i.e.,
    \begin{equation}\label{STYYa}
        \Pr\{ST|Y,Y_a\}=\sum_{b=0}^aC_a^b(\rho)(1-\frac{1}{L})^b.
    \end{equation}
    Substituting \eqref{YaY} and \eqref{STYYa} into \eqref{STY} yields \eqref{TSR}.

\subsection{Proof of Theorem \ref{T1}}\label{appA}
Recall that $U_x$ denotes a variable that represents the consecutive time slots elapsed from the initial entry of the OBO counter of an STA into level $x$ until the subsequent occurrence of the OBO counter reaching $0$ for the first time. 

    Once the OBO counter of an STA enters level $m$, it takes $U_m$ slots for the OBO counter to reach $0$.
    Subsequently, we have two possible cases: \textit{1)} the STA transmits a packet successfully with probability $q$; \textit{2)} the transmission fails with probability $1-q$ and then it takes the STA $R_m$ slots to successfully transmit a packet recursively.
    In this sense, we have
    \begin{equation}\label{ERm}
        \mathbb{E}[R_m]=q\mathbb{E}[U_m]+(1-q)\mathbb{E}[U_m+R_m],
    \end{equation}
    \begin{equation}\label{ERmSq}
        \mathbb{E}[R^2_m]=q\mathbb{E}[U^2_m]+(1-q)\mathbb{E}[(U_m+R_m)^2].
    \end{equation}
    Next, we focus on $U_x$, $x\in\{0,1,\cdots,m-1\}$, which, according to Fig. \ref{MC}, is associated with  the following countdown cases:
    \begin{itemize}
        \item With probability $\frac{L+1}{W_x}$, the count is initialized to state $\left\langle x,0 \right\rangle$, and then $U_x=1$.
        \item With probability $\frac{L}{W_x}$, the count is initialized to one of states $\left\langle x,(\alpha-1)L+1 \right\rangle,\left\langle x,(\alpha-1)L+2 \right\rangle,\cdots,\left\langle x,\alpha L \right\rangle$ for $\alpha\in\{0,1,\cdots,\alpha_x-1\}$. Then, $U_x=\alpha+1$.
        \item With probability $\frac{\beta_x}{W_x}$, the count is initialized to one of states $\left\langle x,(\alpha_x-1)L+1 \right\rangle,\left\langle x,(\alpha_x-1)L+2 \right\rangle,\cdots,$ and\\
        $\left\langle x, (\alpha_x-1)L+\beta_x \right\rangle$.
        Then, $U_x=\alpha_x+1$.
    \end{itemize}
    Considering all three above cases, we can achieve
    \begin{equation}\label{EUm}
    \begin{split}
        \mathbb{E}[U_x]&=\frac{L+1}{W_x}+\sum_{\alpha=2}^{\alpha_x}\frac{\alpha L}{W_x}+(\alpha_x+1)\frac{\beta_x}{W_x}\\
        &=\frac{\alpha_x(\alpha_x+1)}{2}\frac{L}{W_x}+\frac{(\alpha_x+1)\beta_x+1}{W_x},
    \end{split}
    \end{equation}
    \begin{equation}\label{EUmSq}
    \begin{split}
        \mathbb{E}[U_x^2]&=\frac{L+1}{W_x}+\sum_{\alpha=2}^{\alpha_x}\frac{\alpha^2 L}{W_x}+(\alpha_x+1)^2\frac{\beta_x}{W_x}\\
        &=\frac{\alpha_x(\alpha_x+1)(2\alpha_x+1)}{6}\frac{L}{W_x}+\frac{(\alpha_x+1)^2\beta_x+1}{W_x}.
    \end{split}
    \end{equation}
    Substituting \eqref{EUm} into \eqref{ERm}, letting $x=m$, and after some manipulation, we can obtain \eqref{ERmCF}.

    Further, upon the STx's OBO counter entering level $x$ for $0\le x<m$, it requires $U_x$ slots for the counter to reach $0$ and transmit.
    The transmission also succeeds with probability $q$, and fails with probability $1-q$.
    However, if the transmission fails, the next successful transmission would occur after $R_{x+1}$ slots, thereby yielding
    \begin{equation}\label{ERa}
        \mathbb{E}[R_x]=q\mathbb{E}[U_x]+(1-q)\mathbb{E}[U_x+R_{x+1}],
    \end{equation}
    \begin{equation}\label{ERaSq}
        \mathbb{E}[R_x^2]=q\mathbb{E}[U_x^2]+(1-q)\mathbb{E}[(U_x+R_{x+1})^2].
    \end{equation}
    Based on the three countdown cases and the recursiveness of $\mathcal{M}_2$, we can present
    \begin{equation}\label{EU+RSq}
        \begin{split}
            &\mathbb{E}[(U_x+R_{x'})^2]=\mathbb{E}[(1+R_{x'})^2]\frac{L+1}{W_x}\\
            &+\sum_{\alpha=2}^{\alpha_x}\mathbb{E}[(\alpha+R_{x'})^2]\frac{L}{W_x}+\mathbb{E}[(\alpha_x+1+R_{x'})^2]\frac{\beta_x}{W_x}\\
            &=\mathbb{E}[U_x^2]+2\mathbb{E}[U_x]\mathbb{E}[R_{x'}]+\mathbb{E}[R^2_{x'}],
        \end{split}
    \end{equation}
    where $x'=\min(x+1,m)$.
    Substituting \eqref{EUm} into \eqref{ERa}, together with some manipulations, yields \eqref{ERaCF}; similarly, substituting \eqref{EUmSq} and \eqref{EU+RSq} into \eqref{ERmSq} provides \eqref{ERmSqCF}, and substituting \eqref{EUmSq} and \eqref{EU+RSq} into \eqref{ERaSq} results in~\eqref{ERaSqCF}.

%\subsection{Proof of Corollary \ref{AoINoChange}}\label{appC1}
   
\subsection{Proof of Theorem \ref{TUincrease}}\label{appC2}
We calculate the partial derivatives of $\mathbb{E}[U_x]$ given in Theorem \ref{TUincrease} and have
\begin{equation}\label{pd1}
\begin{split}
    \frac{\partial \mathbb{E}[U_x]}{\partial \alpha_x}&\!=\!\dfrac{L^2\alpha_x^2\!+\!\left(2L\beta_x+2L\right)\alpha_x+2\beta_x^2+\left(2-L\right)\beta_x\!-\!L}{2\left(L\alpha_x+\beta_x+1\right)^2}\\
    &>0,
\end{split}
\end{equation} 
for $\alpha_x\ge 1,\beta_x\ge 0$, and
\begin{equation}\label{pd2}
    \frac{\partial \mathbb{E}[U_x]}{\partial \beta_x}=\dfrac{\alpha_x\left(L\alpha_x+L+2\right)}{2\left(\beta_x+L\alpha_x+1\right)^2}\ge 0
\end{equation}
for $\alpha_x\ge 0,\beta_x\ge 0$.
% For $x'>x$ and $\alpha_x\ge 1$, i.e., $W_1>L+1$, we have $W_{x'}>W_x$, implying $\alpha_{x'}>\alpha_x$.
For $x'>x$ and $\alpha_x\ge 1$, i.e., $W_x>L+1$, we have $W_{x'}>W_x$, implying $\alpha_{x'}>\alpha_x$.
Since $\alpha_x$ and $\beta_x$ are the quotient and the remainder of $(W_x-1)/L$, we have
\begin{equation}
\begin{split}
    \mathbb{E}[U_x]&\!=\!\mathbb{E}[U](\alpha_x,\beta_x)\!\overset{(a)}{\le}\!\mathbb{E}[U](\alpha_{x'-1},\beta_x)\!\overset{(b)}{<}\!\mathbb{E}[U](\alpha_{x'-1},L)\\
    &=\mathbb{E}[U](\alpha_{x'},0)\overset{(b)}{<}\mathbb{E}[U](\alpha_{x'},\beta_{x'})=\mathbb{E}[U_{x'}],
\end{split}
\end{equation}
where inequality (a) is by \eqref{pd1} and inequalities (b) are by \eqref{pd2}.
On the other hand, if $W_x\le L+1$, the backoff state of the STA must transit to $\langle x,0 \rangle$ after the STA enters level $x$, implying that $\mathbb{E}[U_x]=1$.
This completes the proof.

\subsection{Proof of Theorem \ref{LB}}\label{appB}
Since $\lambda=1$, the buffer of each STA is always full.
    This indicates $V=0$.
    Furthermore, the evolution of the AoI of an STA can be given by 
    \begin{equation}
    \Delta(t+1)=
    \begin{cases}
        1,& \text{if a status update of the STA is}\\
        & \text{received by the AP in slot $t$,}\\
        \Delta(t)+1,& \text{otherwise,}
    \end{cases}
    \end{equation}
which implies that $S=1$, consistent with substituting $\lambda=1$ into \eqref{approS} under the assumption given in Subsection \ref{EAoIAnaylsis}.
Substituting $V=0$ and $S=1$ into \eqref{AoIUE} presents
\begin{equation}\label{AoISS}
\begin{split}
    \overline{\Delta}&=\frac{\mathbb{E}[K^2]}{2\mathbb{E}[K]}+\frac{1}{2}=\frac{\mathbb{E}[U_0^2]+2(1-q)\mathbb{E}[U_0]\mathbb{E}[R_0]}{2\mathbb{E}[U_0]}+\frac{1}{2}\\
    &\overset{(a)}{=}\frac{\mathbb{E}[U_0^2]}{2\mathbb{E}[U_0]}+\frac{1-q}{q}\mathbb{E}[U_0]+\frac{1}{2}\\
    &\overset{(b)}{\ge}\left(\frac{1}{q}-\frac{1}{2}\right)\mathbb{E}[U_0]+\frac{1}{2},
\end{split}
\end{equation}
where equality $(a)$ is by \eqref{ERm} and inequality $(b)$ is because the fact that the variance of any random number $X$ follows ${\rm Var}[X]=\mathbb{E}[X^2]-\mathbb{E}[X]^2\ge 0$. 

Owing to $\lambda=1$, all STAs in the BSS are active, having
\begin{equation}\label{qSS}
\begin{split}
    q&=\Pr\{\text{An RU is only selected by the STA}\}\\
    &=L\times\frac{1}{L}\left(1-\frac{\rho}{L}\right)^{N-1}=\left(1-\frac{\rho}{L}\right)^{N-1}.
\end{split}
\end{equation}
Recall that for $\lambda=1$, $m=0$ we have
\begin{equation}\label{rhoSS}
    \rho=\frac{W_0}{H_0+W_0}=\frac{\alpha_0L+\beta_0+1}{\frac{L}{2}\alpha_0^2+(\beta_0+\frac{L}{2})\alpha_0+\beta_0+1}=\frac{1}{\mathbb{E}[U_0]}.
\end{equation}
Substituting \eqref{qSS} and \eqref{rhoSS} into \eqref{AoISS}, we obtain \eqref{AoISSo} and \eqref{AoILBSSo}.

% Recalling \eqref{AoISS}, we focus on term $\mathbb{E}[K^2]/\mathbb{E}[K]$.
% Since $U_0\ge 1$, we can conclude that $\mathbb{E}[U_0^2]=\sum_{u=1}^{W_0-1}\Pr\{U_0=u\}u^2$ and $\mathbb{E}[U_0]^2=(\sum_{u=1}^{W_0-1}\Pr\{U_0=u\}u)^2$ are strictly positively correlated to $\mathbb{E}[U_0]=\sum_{u=1}^{W_0-1}\Pr\{U_0=u\}u$, indicating that $\mathbb{E}[U_0^2]$ and $\mathbb{E}[U_0]^2$ are strictly positively correlated.
% Based on this, $\overline{\Delta}$ and $\overline{\Delta}_{LB}$ are positively correlated by equality $(a)$.

\subsection{Proof of Theorem \ref{LBroot}}\label{appC}
Deriving $\tilde{\Delta}_{LB}(W)$, we have
    \begin{equation}
    \begin{split}
        &\frac{\mathrm{d}\tilde{\Delta}_{LB}(W)}{\mathrm{d}W}
        =\left(\dfrac{W^2-L-1}{2LW^2}\right)\\
        &\times\left(-\dfrac{\left(N-1\right)\exp{\left(\frac{N-1}{LU(W)}\right)}}{LU(W)}+\exp{\left(\frac{N-1}{LU(W)}\right)}-\dfrac{1}{2}\right).
    \end{split}      
    \end{equation}

To obtain the real positive roots of $\tilde{\Delta}_{LB}(W)$, we need to find the positive roots of the following functions:
\begin{equation}\label{Wre1}
    f_1(W)=\dfrac{W^2-L-1}{2LW^2},
\end{equation}
%and
\begin{equation}\label{Wre2}
    f_2(W)=-\dfrac{\left(N-1\right)\exp{\left(\frac{N-1}{LU(W)}\right)}}{LU(W)}+\exp{\left(\frac{N-1}{LU(W)}\right)}-\dfrac{1}{2}.
\end{equation}
Clearly, $f_1(W)$ has two roots, given by $W=-\sqrt{L+1}<0$ and $\sqrt{L+1}>0$.
Thus, $\sqrt{L+1}$ is one of the real positive root of $\tilde{\Delta}_{LB}(W)$. 

After some manipulation on $f_2(W)=0$, we have
\begin{equation}
    \left(\dfrac{N-1}{LU(W)}-1\right)\exp{\left(\frac{N-1}{LU(W)}-1\right)}=-\frac{1}{2\mathrm{e}},
\end{equation}
leading to
\begin{equation}\label{Ure}
\begin{split}
    &\frac{N-1}{LU(W)}-1=\mathcal{W}\left(-\frac{1}{2\mathrm{e}}\right)\\
    &\Rightarrow U(W)=\frac{W^2+(L-2)W+L+1}{2WL}=\frac{N-1}{L(\mathcal{W}(-\frac{1}{2\mathrm{e}})+1)}.
\end{split}
\end{equation}

To find the real positive roots of $f_2(W)$, we focus on $U(W)$ in $W\in(0,+\infty)$.
According to the properties of the quadratic equation of one variable, the term $W^2+(L-2)W+L+1$ takes the minimum value at $W=-\frac{L-2}{2}$ and monotonically increases in $[-\frac{L-2}{2},+\infty)$.
We then consider the following two cases: 
\textit{1)} $L=1$ makes the term take the minimum value $\frac{7}{4}$;
\textit{2)} $L\ge 2$ makes the term monotonically increases in $[0,+\infty)$ and equal $L+1$ when $W=0$.
These two facts imply that $U(W)>0$ when $W>0$. 

$\mathcal{W}\left(-\frac{1}{2\mathrm{e}}\right)$ has two real values, given by $\mathcal{W}_0\left(-\frac{1}{2\mathrm{e}}\right)\approx -0.232$ and $\mathcal{W}_{-1}\left(-\frac{1}{2\mathrm{e}}\right)\approx -2.6783$, where $\mathcal{W}_{-1}(x)$ is the Lambert W function of branch $-1$.
Using $\mathcal{W}_{-1}\left(-\frac{1}{2\mathrm{e}}\right)$ yields $\frac{N-1}{L(\mathcal{W}\left(-\frac{1}{2\mathrm{e}}\right)+1)}<0$, indicating that $f_2(W)$ has no positive root.
Hence, we substitute $\mathcal{W}_{0}\left(-\frac{1}{2\mathrm{e}}\right)$ into \eqref{Ure} and then achieve the following quadratic equation:
\begin{equation}\label{Wre3}
    W^2+BW+L+1=0.
\end{equation}
According to the solving rules of quadratic equations, \eqref{Wre3} has no real solution when $B^2-4(L+1)<0$.
Moreover, since $\mathcal{W}_0(-\frac{1}{\mathrm{2e}})$ is a transcendental number, we cannot have $B^2=4(L+1)$, meaning that \eqref{Wre3} cannot have two equal solutions.
Further, \eqref{Wre3} has two real solutions, 
\begin{equation}
    W_1^*=\frac{-B-\sqrt{B^2-4(L+1)}}{2}
\end{equation}
%and
\begin{equation}
    W_2^*=\frac{-B+\sqrt{B^2-4(L+1)}}{2}
\end{equation}
when condition \eqref{C1} is satisfied.
If $B>0$, both $W_1^*,W_2^*<0$ since  
\begin{equation}\label{IE}
    \sqrt{B^2-4(L+1)}<|B|.
\end{equation}
In this context, $\tilde{\Delta}_{LB}(W)$ has only one real positive root $r_2=\sqrt{L+1}$.
Similarly, both $W_1^*,W_2^*>0$ if $B<0$, which is condition \eqref{C2}, by \eqref{IE}.
As such, $r_1=W_1^*$, $r_3=W_2^*$, and $r_2$ are three real positive roots of $\tilde{\Delta}_{LB}(W)$ has only one local minimum point located at $r_2$ in $(0,+\infty)$.

On the other hand, we can present
\begin{equation}
\begin{split}
    \lim_{W\to 0^+}U(W)&=\lim_{W\to 0^+}\frac{W+L-2+\frac{L+1}{W}}{2L}=+\infty\\
    &=\lim_{W\to +\infty}U(W).
\end{split}
\end{equation}
By \eqref{AoIapproLBSS}, we can conclude 
\begin{equation}\label{AoIapproLB0+inty}
    \lim_{W\to 0^+}\tilde{\Delta}_{LB}(W)=\lim_{W\to +\infty}\tilde{\Delta}_{LB}(W)=+\infty.
\end{equation}
In light of this, when condition \eqref{C1}, \eqref{C2} hold, $\tilde{\Delta}_{LB}(W)$ has three stationary points located at the three positive roots, and we have two possible cases in $(0,+\infty)$: 
\begin{enumerate}
    \item\label{FC1} The smallest and largest roots are not extreme points and the remaining root is the local minimum point;
    \item\label{FC2} The smallest and largest roots are two local minimum points, and the remaining root is the local maximum point.
    Besides, the local maximum is larger than the two local minimums.
\end{enumerate} 
Clearly, $\tilde{\Delta}_{LB}(r_1)=\tilde{\Delta}_{LB}(r_3)$, illustrating that only $r_1<r_2<r_3$ can satisfy cases (\ref{FC1}) and (\ref{FC2}).
Furthermore, we have
\begin{equation}
    \lim_{W\to 0^+}f_1(W)=-\infty,\lim_{W\to +\infty}f_1(W)=\frac{1}{2L},
\end{equation}
indicating that $f_1(W)$ is negative and positive in $(0,r_2)$ and $(r_2,+\infty)$, respectively, and
\begin{equation}
    \lim_{W\to 0^+}f_2(W)=\lim_{W\to +\infty}f_2(W)=\frac{1}{2},
\end{equation}
representing that $f_2(W)$ is positive in $(0,r_1)$, $(r_3,+\infty)$ and negative in $(r_1,r_3)$.
Drawing upon this, case (\ref{FC1}) is not true.

Moreover, if condition \eqref{C1}, \eqref{C2} are not both satisfied, $\tilde{\Delta}_{LB}(W)$ has only one positive root $r_2$.
Then, by \eqref{AoIapproLB0+inty}, $\tilde{\Delta}_{LB}$ has only one local minimum point. This completes the proof.

% \begin{thebibliography}{1}
\bibliographystyle{IEEEtran}
\bibliography{Ref}
% \end{thebibliography}

\end{document}